\documentclass[sigconf]{aamas} 
\usepackage{balance} % for balancing columns on the final page

\setcopyright{none}
\acmConference[AAMAS '24]{23rd International Conference
on Autonomous Agents and Multiagent Systems (AAMAS 2024)}{May 6 -- 10, 2024}
{Auckland, New Zealand}{N.~Alechina, V.~Dignum, M.~Dastani, J.S.~Sichman (eds.)}
\acmYear{2024}
\acmDOI{}
\acmPrice{}
\acmISBN{}

\acmSubmissionID{<<EasyChair submission id>>}

\title[AAMAS-2024 Formatting Instructions]{Atlas-X Equity Financing: Unlocking New Methods to Securely Obfuscate Axe Inventory Data Based on Differential Privacy}

\author{Antigoni Polychroniadou}
\affiliation{
  \institution{J.P. Morgan AI Research, AlgoCRYPT CoE }
  \city{New York}
  \country{USA}}
\email{antigoni.polychroniadou@jpmorgan.com}

\author{Gabriele Cipriani}
\affiliation{
  \institution{J.P. Morgan Quantitative Research}
  \city{London}
  \country{United Kingdom}}
\email{gabriele.cipriani@jpmorgan.com}

\author{Richard Hua}
\affiliation{
  \institution{J.P. Morgan Quantitative Research}
  \city{New York}
  \country{USA}}
\email{richard.hua@jpmorgan.com}

\author{Tucker Balch}
\affiliation{
  \institution{J.P. Morgan AI Research}
  \city{New York}
  \country{USA}}
\email{tucker.balch@jpmchase.com}
\begin{abstract}
Banks publish daily a list of available securities/assets (axe list) to selected clients to help them effectively locate Long (buy) or Short (sell) trades at reduced financing rates. This reduces costs for the bank, as the list aggregates the bank's internal firm inventory per asset for all clients of long as well as short trades. However, this is somewhat problematic: (1) the bank's inventory is revealed; (2) trades of clients who contribute to the aggregated list, particularly those deemed large, are revealed to other clients. Clients conducting sizable trades with the bank and possessing a portion of the aggregated asset exceeding $50\%$ are considered to be concentrated clients. This could potentially reveal a trading concentrated client's activity to their competitors, thus providing an unfair advantage over the market.

Atlas-X Axe Obfuscation, powered by new differential private methods, enables a bank to obfuscate its published axe list on a daily basis while under continual observation, thus maintaining an acceptable inventory Profit and Loss (P\&L) cost pertaining to the noisy obfuscated axe list while reducing the clients' trading activity leakage. Our main differential private innovation is a differential private aggregator for streams (time series data) of both positive and negative integers under continual observation.

For the last two years, Atlas-X system has been live in production across three major regions—USA, Europe, and Asia—at J.P. Morgan, a major financial institution, facilitating significant profitability. To our knowledge, it is the first differential privacy solution to be deployed in the financial sector. We also report benchmarks of our algorithm based on (anonymous) real and synthetic data to showcase the quality of our obfuscation and its success in production.
\end{abstract}

\newcommand{\BibTeX}{\rm B\kern-.05em{\sc i\kern-.025em b}\kern-.08em\TeX}

\usepackage{amsmath,amsfonts}
\usepackage{algorithmic}
\usepackage{graphicx}
\usepackage{textcomp}
\usepackage{xcolor}
\def\BibTeX{{\rm B\kern-.05em{\sc i\kern-.025em b}\kern-.08em
    T\kern-.1667em\lower.7ex\hbox{E}\kern-.125emX}}
\setcitestyle{numbers}

\usepackage{hyperref}

\DeclareMathOperator{\sgn}{sgn}
\usepackage{pifont}

\newcommand{\Mod}[1]{\ (\mathrm{mod}\ #1)}

\usepackage{tikz}
\usepackage[utf8]{inputenc}
\usepackage{enumitem}
\usepackage{amsmath}
\usepackage{graphicx}
\usepackage{amsfonts, xcolor}

\usepackage{fancyhdr,graphicx,amsmath}
\usepackage[ruled,vlined]{algorithm2e}
\usepackage{booktabs}

\newtheorem{theorem}{Theorem}[section]

\newtheorem{definition}[theorem]{Definition}

\newtheorem{corollary}[theorem]{Corollary}

\makeatletter
\gdef\@copyrightpermission{
	\begin{minipage}{0.3\columnwidth}
		\href{https://creativecommons.org/licenses/by/4.0/}{\includegraphics[width=0.90\textwidth]{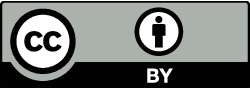}}
	\end{minipage}\hfill
	\begin{minipage}{0.7\columnwidth}
		\href{https://creativecommons.org/licenses/by/4.0/}{This work is licensed under a Creative Commons Attribution International 4.0 License.}
	\end{minipage}
	\vspace{5pt}
}
\makeatother

\begin{document}

%%% The following commands remove the headers in your paper. For final 
%%% papers, these will be inserted during the pagination process.

\pagestyle{fancy}
\fancyhead{}

%%% The next command prints the information defined in the preamble.

\maketitle 

%%%%%%%%%%%%%%%%%%%%%%%%%%%%%%%%%%%%%%%%%%%%%%%%%%%%%%%%%%%%%%%%%%%%%%%%

\section{Introduction}

An axe is an interest in a particular security that a firm is looking to buy or sell\footnote{The term comes from the jargon: ``having an axe to grind''.}. In general, a firm providing an axe to external counterparties has a strong interest in keeping such information private as it provides an indication of the direction (buy or sell) they want to trade a particular security. If other market participants are informed of how a particular firm is axed in a given security, they can extract precious information on the firm's trading strategy and, perhaps, could even drive the price of the security in a more disadvantageous direction before the firm can transact. 

Large broker-dealer banks, including J.P. Morgan, distribute aggregated axe lists to clients (called hedge funds) with the aim of reducing the costs of running their activities. The axe list is shared electronically (via email or other means) and, most importantly, is common to all clients receiving it. It consists of a list of tuples $({\sf symb}, {\sf direction},{\sf quantity})$ where $\sf symb$ is the symbol of the security/asset to buy or sell based on the ${\sf direction}$ and $\sf quantity$ is the number of the shares (positions) of the security to trade. The universe of assets covered by the axe list is rather large, encompassing thousands of securities listed in major markets. For a given asset, the bank's axe is given by the aggregation of the positions held by the bank. Importantly, when the traded positions of a large (``concentrated'') client contribute the most to the axe quantity, the published axe reveals the trading activity of the client. This is particularly problematic because hedge funds' trading strategies are confidential and their disclosure can undermine the funds' performance. Sensitive trading moves reconstruction is feasible due to various factors, including side information. Clients report positions exceeding specific thresholds to regulatory bodies like the Fed, providing a trail for piecing together trading patterns. Informal conversations in financial circles can also divulge valuable insights. This scenario poses the risk of someone replicating trading moves by observing aggregations, impacting the original strategy's efficiency and market dynamics.

The problem we address in this work is how to minimize the adverse effects of the information leakage caused by sharing the axe list with clients.
Such undesirable consequences are important both from a reputational point of view, with the bank losing clients which don't want their trading decision made public, as well as a risk management / financial one when the information contained in the axe list is used to trade against the bank itself. Such problems are also exacerbated by the fact that the axe list is published daily, with clients having access to the full history of the published axe lists.

Historically, banks have been using some ad-hoc methods to mitigate the leakage. For instance, they might aggregate several stocks together into buckets (e.g., reveal only range of available stocks to trade in some sector), or trim the positions of other stocks. This does not guarantee privacy and does not provide a useful axe list with good utility (in the case, the profit for the bank). In some cases the clients’ positions are removed from the inventory to eliminate the leakage at the expense of poor utility for the bank and the inability of the bank to offer reduced rates to clients. 

Our approach is instead more robust, based on a new differential private aggregator for data streams under continual observation. We also introduce precise measures to quantify the utility of the published axe (in terms of profits for the bank) as well as the quality of the obfuscation.

%Please see Section \ref{FinancingActivities} for a discussion of the services provided by a Prime desk and, specifically, the rationale behind the use of axe lists in such a context. 

\section{Problem Statement}

In this paper, we investigate an intriguing question related to the secure sharing of aggregated time-series data (e.g. axe inventory trades) on a daily basis of all clients, while preserving the privacy of any changes in the direction (buy or sell) of the data/trades of contributing concentrated clients’ trading activity.\\

\emph{
How can a bank maintain the continuous release of updated aggregate time-series market data while preserving each individual client’s privacy?}\\

As illustrated in Figure \ref{DPh}, by comparing two different axe lists - one that includes the concentrated clients' positions with the bank and one that does not - we would like to obscure whether a single concentrated client is buying or selling an asset on any given day. If the direction is revealed, then the client's activity is exposed, as they are guiding the direction of the asset in question. As depicted in Figure~\ref{axeexample}, a concentrated client (represented by the green line) holds the majority of positions on the true axe list (represented by the blue line), dictating daily movements. Our objective is to obscure the true axe (represented by the orange line) by concealing the directional activities of the concentrated client, whether it involves buying or selling (increasing or decreasing quantity/ positions). If all the clients follow the same direction as the concentrated client, then this is normal behavior that we should not attempt to hide since the client's behavior is not particularly distinctive and follows the crowd. We would like to maintain privacy with respect to a utility constraint. Jumping ahead, the utility is determined by the profits or losses the bank can incur by obfuscating the direction. 

\begin{figure}
  \centering
  \includegraphics[width=0.5\textwidth]{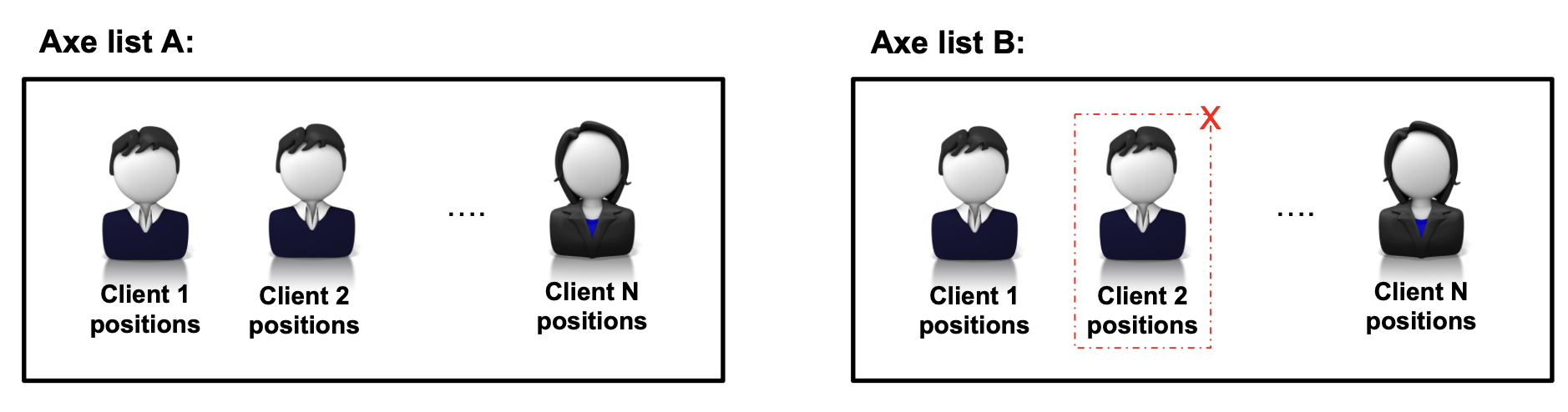}
  \caption{The daily direction (Buy or Sell) of two axe lists that differ only in the positions of a single concentrated client should be statistically indistinguishable.}\label{DPh}
\end{figure}

\begin{figure}
  \centering
  \includegraphics[width=0.45\textwidth]{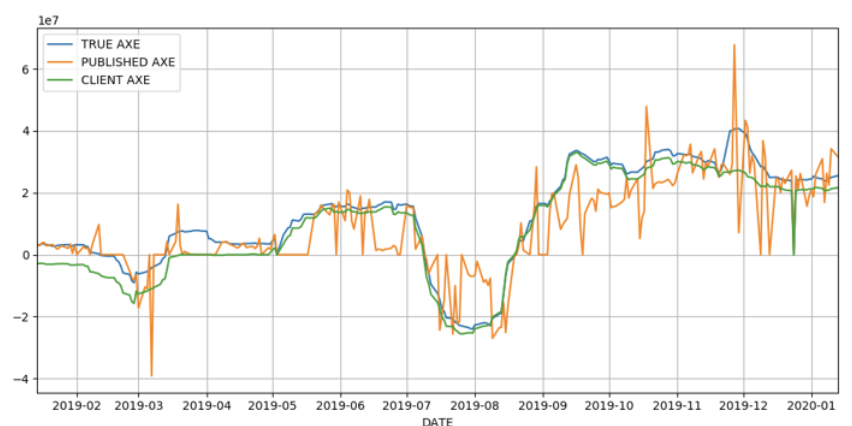}
  \caption{ Example of an obfuscated published axe (in orange color) for a given asset, together with the historical data for the bank's true axe (in blue color) and the positions of a highly concentrated client (in green color). The Y-axis refers to the axe quantity while the X-axis the observation date.}\label{axeexample}
\end{figure}

This problem calls for differential privacy. The notion of differential privacy was proposed by Dwork et al.~\cite{dwork2006our}. Since then, there is an extensive work in the literature studying the different tradeoffs between utility and privacy. However, the differentially private setting we consider is different from the traditional setting which assumes a static input database, and a third party that needs to publish some obfuscated/sanitized aggregate statistics of the database once. In our use case, the database is dynamic and changes every day. The differential private  mechanism needs to update the published aggregate statistics on a daily basis. Traditional differential private mechanisms can lead to a significant loss in terms of utility or privacy. 

In the context of differential privacy we address the following generic question: \\

%\emph{
%Is it possible to publish a dataset daily and reduce the leakage pertaining to the presence or absence of an individual in the dataset while preserving the dataset’s utility? }\\
\emph{
How can data aggregators continuously release updated aggregate statistics while preserving each individual user’s privacy and without degrading the utility of the data? 
}\\

The utility measure that we have identified and consider in this work is the cost of keeping positions on the bank's inventory. We describe the P$\&$L as the cost of financing a position on a given asset, due to the need of raising the cash to buy it or acquiring the asset to sell it in the market.

%The utility measure that we consider in this work is the cost of keeping ``non-internalized'' positions on the bank's inventory - balance sheet. More details (on the financial terms) are given in Section \ref{FinancingActivities} and Section \ref{SecPLMetric} where we describe the cost of financing a position on a given asset, due to the need of raising the cash to buy it (``long'' position) or acquiring the asset to sell it in the market (``short'' position).

The work of~\cite{DworkNPR10} proposes algorithms for differential privacy under continual observation but only for the minimum functionality of counting binary ($0$ or $1$) values at each time-step. The counter statistic is a basic primitive in numerous data streaming algorithms. In this work, we are interested in more complex statistics released under continual observation and explore the utility of the use case at  hand.  

In this section, for readability reasons we have abstracted and generalized the problem statement based on the use case of axe inventory. For an in-depth connection of the problem statement to the axe inventory use case, please refer to Appendix~\ref{FinancingActivities}, wherein we have also defined all the financial terms related to the problem.

\subsection{Our Contributions} 

Our contributions can be succinctly summarized as follows:

\begin{itemize}
   \item {\bf Real-World Use-Case Identification}: We pinpoint a real-world scenario that underscores the potential benefits of DP.  Atlas-X is the first DP solution under CO running live in the financial arena.
    \item {\bf New algorithm for DP under continual observation (CO)}: We introduce a new algorithm on the harder setting of DP with CO;  DP under CO presents greater challenges than traditional DP due to the persistent data access. While traditional DP focuses on static datasets. Key Challenges:
    \begin{enumerate}
        \item Accumulative Privacy Risk: With ongoing data addition and observations, the cumulative privacy risk grows. New data points incrementally reveal more about individuals, raising the risk of reidentification and sensitive information exposure.
        \item Cumulative Knowledge: In CO, adversaries might exploit the accumulated knowledge gained from earlier queries to deduce more sensitive information.
    \end{enumerate}              
    These challenges require more advanced techniques and strategies to ensure long-term privacy preservation while enabling valuable insights from the continuously-observed data.
\item {\bf Analysis $\&$ Implementation:} Identified the metrics/risks studied for this use case together with the business and provided an analysis based on real production data. Furthermore, we report benchmarks of the Atlas-X system, which runs in production. The system offers increased opportunities for clients to locate trades at advantageous prices as well as achieving better profits for the bank. Atlas-X has also proven to be useful in retaining existing clients of the bank, as we are able to prove that information about their trading activity is effectively safeguarded.

\item {\bf Integration with Trading Platform:} Our system seamlessly integrates into an existing trading platform of J.P. Morgan, further validating its practical utility. %The production system operates on a Linux machine based on python 3.7 (256GB memory) within a universe of 20,000 stocks and 600 clients without scalability concerns. 

\end{itemize}

Next we provide more details on our new algorithm and the novel use case.

\noindent{\bf New Algorithm:} We show how to address the above questions using differential privacy techniques. We propose a differentially private continual aggregator that outputs at every time step the approximate updated aggregator. We can achieve a construction that has error that is only poly-log in the number of time steps. Assume that the input stream $\sigma \in \mathbb{Z}$ is a sequence of positive and negative integers. The integer $\sigma(t)$ at time $t \in \mathbb{N}$ may denote whether positions/shares in a stock increased or decreased at time $t$, e.g., whether a client bought or sold $\sigma$ shares of a stock at time $t$. The mechanism must output an approximate aggregator of the sum of all positive and negative integers seen so far until timestep $t$. We propose an $\epsilon$-differentially private continual aggregator with small error. For each $t \in \mathbb{N}$ we guarantee $O(\frac{ T^{1/4}\sqrt{\Delta}}{\epsilon})$ error with global sensitivity $\Delta$. See Theorem~\ref{them} for our formal statement and Theorem~\ref{them2}. Prior works~\cite{DworkNPR10,ChanSS11} have considered simpler statistics under continual observation as well as  simpler utility considerations.

\noindent{\bf Use Case:} We have identified a real-world problem for differential privacy under continual observation on a large dataset in which the privacy of the previous axe inventory publication can be significantly enhanced. We propose a new privacy preserving algorithm that generates a noisy axe list while protecting clients’ privacy and maintaining the desired profit for the bank (P$\&$L). Differential privacy is a statistical learning tool that enables us to add carefully computed mathematical noise to the axe list. The noise term is large enough to obfuscate individual client positions and small enough to achieve the desired P$\&$L. Our new algorithm is robust with provable guarantees of privacy. To estimate the effectiveness of our method, we have also defined the utility associated with the axe publication as well as measures of the quality of the obfuscation. The model parameters have then been derived by employing these findings, as displayed in Section \ref{sec:exp} for further detail.

\subsection{Related work}

The works of \cite{DworkNPR10,ChanSS11} introduced the concept of differential privacy under continual observation and constructed differentially private continual counters of streams of $0$'s and $1$'s. Their binary mechanisms are used in the context of the orthogonal problem of privacy-preserving federated learning~\cite{BonawitzIKMMPRS17} with the most recent ones being~\cite{DBLP:conf/icml/KairouzM00TX21} and it represents a separate context from our primary use case where differential privacy under continuous observation is used. Privacy-preserving federated learning is a distributed machine learning approach that allows multiple parties to collaboratively train a shared model while keeping their data private. It is an emerging technology that is gaining popularity due to its ability to protect data privacy and reduce data movement while allowing multiple parties to train a model with their own data. To ensure differential privacy, federated learning employs various clipping mechanisms~\cite{DBLP:conf/ccs/AbadiCGMMT016,DBLP:conf/iclr/McMahanRT018,DBLP:conf/nips/AndrewTMR21} too. The latest advancements in privacy-preserving federated learning~\cite{GuoPSBB22,acorn,LiLPT23,MaWAPR23} based on secure multiparty computation (MPC) provide enhanced security measures by employing masked or encrypted training gradients.

Prime Match~\cite{PolychroniadouA23} from J.P. Morgan, based on MPC, significantly enhances security and privacy. In Prime Match, buy and sell orders of clients and the bank are encrypted for matching, with orders only being revealed if a match occurs. Unlike Atlas-X based on differential privacy, which hides specific axe dataset/order properties without requiring client participation in MPC, Prime Match ensures no information leakage unless a match is found.

\subsection{Technical Overview} 

\noindent{\bf Problem Statement.} Our goal is to achieve a differential private mechanism for aggregation under continual observation. A mechanism is differentially private if it cannot be used to distinguish two streams that are almost the same. In other words, an attacker cannot tell whether an event of interest took place or not by looking the output of the mechanism over time. For example, the adversary is unable to determine whether a concentrated clients’s positions are included in the inventory axe list at some time $t$. 

We abstract the problem as follows: we consider streams of positive and negative numbers. Let $\sigma(t)$ be an item in the stream at time $t \in \mathbb{N}$ which is either a positive or negative integer. At every time $t$, we wish to output the sum of numbers $\alpha(t)=\sum_{i}^t \sigma(t)$, the aggregator, that have arrived up to time $t$ from $i=1$. 

\noindent{\bf Naive mechanism} In this mechanism at every time step $t$, the mechanism answers with a new sum, and randomizes the answer with fresh independent noise i.e. $\alpha(t)+ noise$ where $\alpha(t)$ is the true aggregator at timestep $t$. The drawback is that the privacy loss grows linearly with respect to the number of
queries, which is $T$ in our setting. $T$ is an upper bound on time. 

\noindent{\bf Simple mechanism} Another approach is to add independent noise to each item of the stream, i.e. $\sigma(t)+ noise_t$ and the mechanism outputs $\sum_{i\leq t} (\sigma(t)+ noise_t)$ at time $t$. In this case the privacy loss depends on $\sqrt{T}$. 

\noindent{\bf Window mechanism} In this mechanism we want to publish noisy versions of some partial sums as new items
arrive. Given the partial sums, an observer computes an estimate for the aggregator at each time step by summing up an appropriate selection of partial sums. For instance, in the naive mechanism, $\alpha(t)$ can be seen as a sum of noisy partial sums where each item $\sigma(t)$ appears in $O(T)$ of these partial sums and this is why the privacy loss is linear in $T$. In particular, when an item is flipped in the incoming stream, O(T) of the partial sums will be affected. In the simple mechanism,  the published partial sums are noisy versions of each item. That said, each item appears in only one partial sum but each aggregator is the sum of $O(T)$ partial sums. 

To guarantee small privacy loss, we would like to have each item appear in a small number of partial sums. Moreover, to achieve smaller error we want each aggregator to be a sum of a small number of partial sums since the noises add up as we sum up several noisy partial sums.  Inspired by the work of~\cite{DworkNPR10} who consider a counter mechanism for counting an incoming stream of only $0$s and $1$s, we group consecutive items contiguous windows of size $B$. Then the idea is that within a window, we apply the simple mechanism from above. Then, treating each window as a single item we apply again the simple mechanism. More details of our algorithm are given in Algorithm 1. Jumping ahead, each $\alpha(t)$ and $\beta(d(t))$ is a noisy partial sum and one can reconstruct the approximate aggregator at any time step from these noisy partial sums. This Algorithm achieves $\epsilon$-differential privacy and $O(\frac{ T^{1/4}\sqrt{\Delta}}{\epsilon})$ error where $\Delta$ is our dynamic global sensitivity -- which we calculated based on the windows.  

\noindent{\bf Binary mechanism} In Algorithm 2 we show that the error can be reduced to logarithmic in the number of time steps using the so called binary mechanism. The idea is that the grouping of the items depends on the binary representation of the number $t$. Consider a binary interval tree, a partial sum corresponding to each node in the tree is published. To reconstruct the current aggregator it suffices to find a set of nodes in the tree to uniquely cover the time range from 1 to $t$. In this case, every time step appears in $O(\log T)$ partial sums and every aggregator can be represented with a set of $O(\log T)$ nodes.

\noindent{\bf Challenges for the Axe Inventory Obfuscation} In Section~\ref{SecObfuscationMetrics} we carefully and formally define our obfuscation metrics for our use case. As it is described in Section~\ref{SecPLMetric}, our Axe obfuscation algorithms should aim at publishing noisy obfuscated axes that are not too different from the true one, as failure to do so can cause a P$\&$L loss for the bank. However, at the same time we define the leakage probability in  Section \ref{SecLeakageProbability} which is a metric to indicate whether increased or decreased positions to an asset are not
observable in the published noisy axe.

\section{Preliminaries} 

\subsection{Differential Privacy}

Differential privacy~\cite{dwork2006our} states that if there are two databases that differ by only one element, they are statistically indistinguishable from each other. In particular, if an observer cannot tell whether the element is in the dataset or not, she will not be able to determine anything else about the element either. 

\begin{definition}
    ($\epsilon$-differential privacy  \cite{dwork2006calibrating}) For any two neighboring datasets $\mathcal{D}_1\sim\mathcal{D}_2$ that differ by one element, a randomized mechanism $\mathcal{A}$: $\mathcal{D}\rightarrow \mathcal{O}$ preserves $\epsilon$-differential privacy ($\epsilon$-DP) when there exists $\epsilon > 0$ such that,
    \begin{equation}
    \text{Pr [}\mathcal{A}(D_1)\in \mathcal{T} \text{]}\leq e^\epsilon \text{ Pr [}\mathcal{A}(D_2)\in \mathcal{T}\text{]}    
    \end{equation}
    holds for every subset $\mathcal{T} \subseteq \mathcal{O}$, where $\mathcal{D}$ is a dataset, $\mathcal{T}$ is the response set, and $\mathcal{O}$ depicts the set of all outcomes.
\end{definition}

The value $\epsilon$ is used to determine how strict the privacy is. A smaller  $\epsilon$ gives better privacy but worse accuracy. Depending on the application $\epsilon$ should be chosen to strike a balance between accuracy and privacy. 

\begin{definition}
    (Global Sensitivity \cite{dwork2006calibrating}) For a real-valued query function $q: \mathcal{D} \rightarrow \mathbb{R}$, where $\mathcal{D}$ denotes the set of all possible datasets, the global sensitivity of $q$, denoted by $\Delta$, is defined as
    \begin{equation}
        \Delta = \max_{\mathcal{D}_1\sim\mathcal{D}_2} |q(\mathcal{D}_1)-q(\mathcal{D}_2)|,
    \end{equation}
    for all $\mathcal{D}_1\in\mathcal{D}$ and $\mathcal{D}_2\in\mathcal{D}$ .
\end{definition}

\subsubsection*{Laplacian Mechanism}
One of the most well-known techniques in differential privacy is the Laplacian mechanism, which uses random noise $X$ drawn from the symmetric Laplacian distribution.  The zero-mean Laplacian distribution has a symmetric probability density function $f(x)$ with a scale parameter $\lambda$ defined as: 
\begin{equation}
    f(x) = \frac{1}{2\lambda}e^{-\frac{|x|}{\lambda}}.
\end{equation}
Given the global sensitivity, $\Delta$, of the query function $q$, and the privacy parameter $\epsilon$, the \textit{Laplacian mechanism} $\mathcal{A}$ uses random noise $X$ drawn from the Laplacian distribution with scale $\lambda = \frac{\Delta}{\epsilon}$. The Laplacian mechanism preserves $\epsilon$-differential privacy \cite{dwork2006our}. 

In our algorithms, the noise may not come from a single Laplace
distribution, but rather is the sum of multiple independent Laplace distributions. The sum of independent Laplace distributions maintains differential privacy~\cite{dwork2006our,ChanSS11}.

\begin{corollary}\label{cor}
 Suppose $\theta_i$’s are independent random variables, where each $\theta_i$ has Laplace distribution $Lap(b_i)$ and suppose $Y=\sum_i \theta_i$ for $i\in [t]$. The quantity $|Y|$ is at most $O(\sqrt{\sum_i b_i^2}\log\frac{1}{\delta})$. We use the following property of the sum of independent Laplace distributions. 
\end{corollary}

\section{Financial Concepts}\label{sec:finterms}

In the supplementary material (see Section \ref{FinancingActivities}) we introduce the financial concepts and jargon used for the real use case in the the paper. A concise summary of the axe inventory use case is that the publishing bank aggregates its internal firm inventory of long (buy) and short (sell) trades and then provides these offerings to its customers in order to equalize their long and short aggregated positions with regard to a given asset.

\noindent{\bf Outline of Section~\ref{FinancingActivities}:} We define the profit and losses incurred by ``long'' and ``short'' positions (which refer to buying and selling respectively, and are defined in detail below), highlighting the importance of hedging costs for the the bank (``funding'' and ``borrow'' rates). We then demonstrate how banks reduce their hedging costs via a process known as ``internalization'' and how we can reduce such costs by enticing clients to trade via axe lists. Lastly, we describe the implications of sharing axe lists among clients and how axe lists can leak information about the trading activity of clients with large (``concentrated'') positions.\\

\section{Obfuscation Metrics}\label{SecObfuscationMetrics}

In this section, we present the obfuscation metrics used to measure the quality of our differentially private method. We have also implemented a monte carlo simulation engine to estimate such metrics, see supplementary material Section~\ref{SectionAxeSimulation} for the details.

\subsection{P\&L}\label{SecPLMetric}
Here we describe an approximation for the daily inventory P\&L realized when an axe trade is executed with a client. It should be noted, as before, that while our approximation discards some aspects of the netting process and other costs, it is anyway a faithful representation of the true P\&L impact.

Referring to Sec. \ref{FinancingActivities}, the borrow/funding P\&L accrued over one day when keeping a net quantity $x(t)$ of a given asset on balance sheet can be summarized as follows:
\begin{eqnarray}
PL^{INV}(t) = \begin{cases}
-r_F(t) x(t) P(t) & \textrm{if }  x(t)  \ge 0 \\
 r_B(t) x(t) P(t) & \textrm{if }  x(t) < 0 
 \end{cases}
 \label{BSpln}
\end{eqnarray}
where $P(t)$ is the asset's price, $x(t)$ the net inventory positions and $r_F$ / $r_B$ the funding / borrow rates, respectively. From now on we'll restrict our analysis to such P\&L contributions, indicating them as ``inventory P\&L'' or, simply, P\&L. 
Equation \ref{BSpln} can be obtained from Equations \ref{LongPL} and \ref{ShortPL} assuming that all risky P\&L contributions are perfectly hedged and setting $t_E - t_S = 1$. It should be noted, nevertheless, that the total P\&L accrued by the bank is affected by other factors like, for instance, the P\&L deriving from fees charged to clients or other considerations.

When a client accepts an axe trade for a quantity $a_{HIT}(t)$, the net inventory changes from $x(t)$ to $x(t) + a_{HIT}(t)$ and the trade's marginal inventory P\&L, which is the P\&L with the axe trade minus the P\&L without, for the bank reads:
\begin{eqnarray}
\Delta PL^{AXE}(t) &=& PL^{INV}\big(a_{HIT}(t) +x(t)\big)  \\ 
&-&PL^{INV}\big(x(t)\big) \notag  \\ 
&=& PL^{INV}\big(a_{HIT}(t) - a_{TRUE}(t)\big) \notag  \\ 
&-&PL^{INV}\big(- a_{TRUE}(t)\big) \notag
\end{eqnarray}\label{AxePnLFormula}
where $a_{TRUE}(t) = -x(t)$ is the ``true'' axe of the positions in inventory. The marginal P\&L measures the effect of a trade on the bank's inventory P\&L: a large positive marginal P\&L is associated with a good trade from a P\&L perspective. By the same token, trades with negative marginal P\&L would cause a P\&L loss for the bank if executed. 

The behaviour of the marginal P\&L as a function of the traded axe quantity $a_{HIT}(t)$ is rather intuitive, see Fig. \ref{AxePnL}:
\begin{itemize}
\item[-] The maximum is achieved when the traded axe quantity is equal to the true one, $a_{HIT}(t) = a_{TRUE}(t)$. That corresponds to the perfect  case in which the axe trade fully consumes the balance sheet, flattening it to zero, hence removing all inventory costs.
\item[-] When the traded axe $a_{HIT}(t)$ is larger then the true axe $a_{TRUE}(t)$, the marginal P\&L decreases at the funding rate $r_F(t)$. That is because any additional axe quantity will make the net inventory longer, hence increasing the funding costs.
\item[-] Otherwise, the marginal P\&L increases at the borrow rate $r_B(t)$. Any smaller traded axe quantity makes the net inventory shorter and increases the borrowing costs.
\end{itemize}

It should be noted that an axe trade generates a profit only when its quantity stays close enough to the true axe quantity:
\begin{eqnarray}\label{PLBounds}
a_{HIT}(t) \in \begin{cases} 
\Big[0, a_{TRUE}(t) \Big(1 + \frac{r_B(t)}{r_F(t)}\Big)\Big] &\textrm{if }  a_{TRUE}(t) \ge 0 \\ \\
\Big[a_{TRUE}(t) \Big(1 + \frac{r_F(t)}{r_B(t)}\Big), 0\Big] &\textrm{if }  a_{TRUE}(t) < 0
\end{cases}
\end{eqnarray}
The intervals above identify the values of $a_{HIT}(t)$ making the marginal P\&L $\Delta PL^{AXE}(t)$ positive.
As a consequence, axe obfuscation algorithms should aim at publishing obfuscated axes that are not too different from the true one, as failure to do so can cause P\&L losses for the bank.

\begin{figure}
  \centering
  \includegraphics[width=0.5\textwidth]{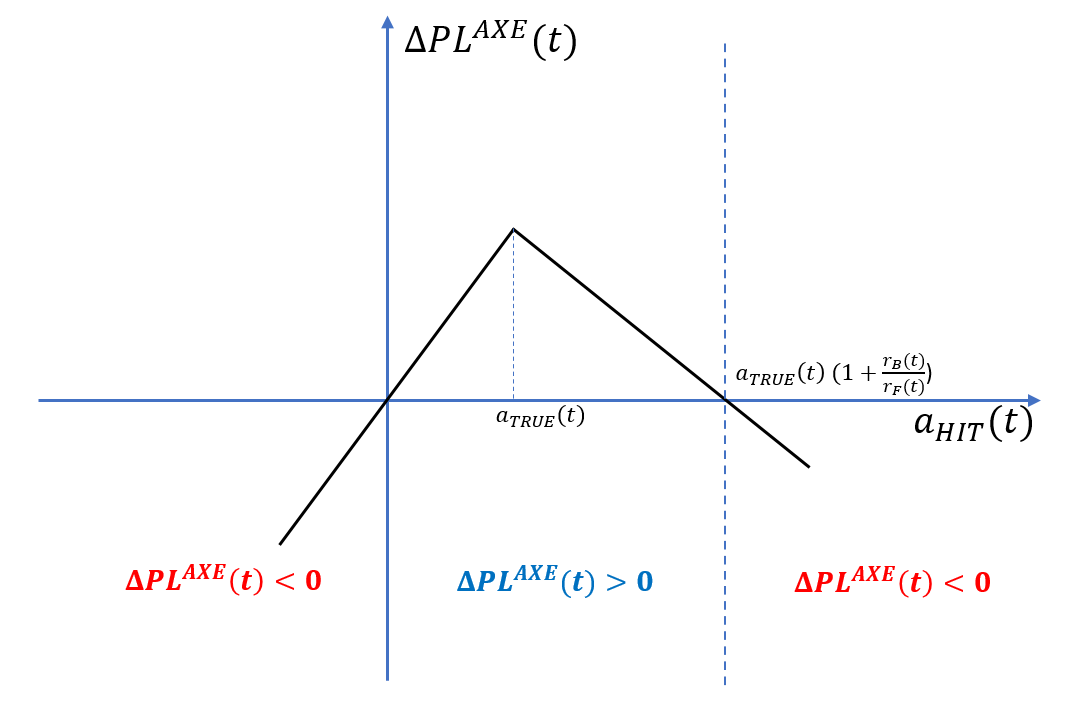}
  \caption{Marginal P\&L profile (Y-axis) of a long axe trade as a function of the axe quantity traded by a client (X-axis).}\label{AxePnL}
\end{figure}

\subsection{Leakage Probability}\label{SecLeakageProbability}
We define the Leakage Probability as the probability of correctly guessing a client trading direction (i.e. whether a given fund is increasing or decreasing its positions on given assets) using the direction of change of the published axe, i.e.:
\begin{eqnarray}
LP(t) = Prob\Big[\sgn \Big( \big(p(t)-p(t-\tau)\big)\\ \big(a_{PUB}(t)-a_{PUB}(t-\tau)\big) \Big) < 0 \Big]  \notag
\end{eqnarray}
where $a_{PUB}(t)$ is the published obfuscated axe list, $p(t)$ is the client's position in the asset and $\tau$ a time lag (a few days, typically).
Please notice that when the client increases their positions ($p(t)-p(t-\tau) > 0$) the effect on the axe is the opposite ($a_{PUB}(t)-a_{PUB}(t-\tau) < 0$), and viceversa hence the definition above.

From a practical point of view, both the direction and quantity of the change in the published axe are important. Our definition of Leakage Probability keeps into account only the direction because it is meant to be a simple and ”robust” estimator of the information leaked
by the published axe. Any attacker able to detect both the sign and the quantity of change in the true axe will also be able to infer the sign only.
A high Leakage Probability denotes a situation in which an attacker could understand whether the bank (or a concentrated client) are taking new positions on. A low Leakage Probability means, instead, that the bank trading decisions (whether they have been increasing or decreasing their exposure to an asset) are not observable in the published axe.

We also define the over-axe frequency and worst case cost in the supplementary material Section~\ref{SectionAxeSimulation}. The over-axe measures how often the published axe, if fully accepted by clients, would cause a negative inventory P\&L / loss for the bank

% \pagebreak
\section{Axe Obfuscation via Continual Aggregator DP Mechanism}\label{sec:algo}

We consider streams of positive and negative numbers. Specifically, $\sigma(t)$ at time $t \in \mathbb{N}$ denotes a positive or a negative integer. At every time $t$, we wish to output the sum of numbers that have arrived up to time $t$. 

\begin{definition}
(Continual Aggregator) Given a stream $\sigma$ of positive and negative integers, let $\sigma(t)\in \mathbb{Z}$ be an integer at time step $t\in \mathbb{N}$ and let: \begin{eqnarray*}
\sigma^{+}(t) &=& \sigma(t) \textrm{ if } \sigma(t) \ge 0 \textrm{ else } 0, \\
\sigma^{-}(t) &=& \sigma(t) \textrm{ if } \sigma(t) \le 0 \textrm{ else } 0,
\end{eqnarray*} the aggregator for the stream is a mapping $\alpha : \mathbb{Z} \rightarrow \mathbb{Z}$ such that for each $t \in \mathbb{N}$, $\alpha(t) := \sum_{i=1}^t \sigma^{+}(i) + \sum_{i=1}^t \sigma^{-}(i)$.
\end{definition}

 Next, we define the notion of a continual aggregator mechanism which continually outputs the sum of integers seen thus far.

\begin{definition}(Continual Aggregator Mechanism)
A counting mechanism $M$ takes a stream $\sigma$ of integers in $\mathbb{Z}$ and produces a (randomized) mapping $M(\sigma): \mathbb{Z} \rightarrow \mathbb{Z}$. Moreover, for all $t \in \mathbb{N}$, $M(\sigma)(t)$ at timestep $t$ is independent of all $\sigma(i)$’s for $i > t$. 
\end{definition}

\begin{definition}(Utility)
An aggregator mechanism $M$ is $(\lambda, \delta)$-useful at time $t$, if for any
stream $\sigma$, with probability (over the randomness of $M$) at least $1 - \delta$, we have $|\alpha(t) -
M(\sigma)(t)| \leq \lambda$. Note that $\lambda$ may be a function of $\delta$ and $t$.
\end{definition}

The above definition covers the usefulness of the mechanism for a single time step. A standard union bound argument can be used for multiple time steps. 

\noindent{\bf Sensitivity:} By using clipping to bound the sensitivity of our summation queries, we are able to enforce upper $\max$ and lower $\min$ bounds on the positions. This ensures that all positions will be below the upper bound, and the resulting sensitivity of a summation query is equal to the difference between the upper and lower bounds used in clipping, $\max-\min$ over a period of time $T$, denoted as either $\max_T-\min_T$ or, for brevity, $\Delta$ in the rest of the paper. We do not choose our clipping bounds by looking at the data; instead, we calculate them by exploiting a property of the dataset that can be determined without viewing the data, thereby providing us with prior knowledge about the scale of the data for clipping. The property refers to the Average Daily Trading Volume (ADTV) of each stock which helps us determine the bounds. ADTV is the average number of shares traded within a day for a given stock, calculated by taking the total number of shares traded over a period of time and dividing it by the number of days in that period. As a rule of thumb in our use case the daily added positions of a client is never above the ADTV (our $\max$) and this is public knowledge, i.e. the bank forbids trades exceeding the Average Daily Trading Volume (ADTV). It is forbidden by the bank since the bank does not want to take any risk executing such large daily trades with (concentrated) clients. We would like to note that ADTV can help confirm trends and patterns to  market participants, which is public information that we do not seek to hide.

\noindent{\bf Our Window Algorithm:}
We concentrate on obfuscating the already clipped True Axe's changes over 1-day periods.
Given a time-grid $\{t=0\ldots T-1\}$, define the True Axe differences as:

\begin{eqnarray*}
\sigma(t) = a_{TRUE}(t) - a_{TRUE}(t-1)
\end{eqnarray*}

Then split them into positive and negative parts:
\begin{eqnarray*}
\sigma^{+}(t) &=& \sigma(t) \textrm{ if } \sigma(t) \ge 0 \textrm{ else } 0 \\
\sigma^{-}(t) &=& \sigma(t) \textrm{ if } \sigma(t) \le 0 \textrm{ else } 0 \\
\end{eqnarray*}

We can reconstruct the True Axe as:
\begin{eqnarray*}
a_{TRUE}(t) &=& a_{TRUE}(0) + \sum_{i=1}^t \sigma^{+}(i) + \sum_{i=1}^t \sigma^{-}(i)
\end{eqnarray*}

Perturbations are applied to the Axe differences ($\theta^{+}(t)$ and $\theta^{-}(t)$ are random shocks we describe in a moment):
\begin{eqnarray*}
\alpha^{+}(t) &=& \sigma^{+}(t) + \theta^{+}(t) \\
\alpha^{-}(t) &=& \sigma^{-}(t) + \theta^{-}(t) \\
\end{eqnarray*}

%We add noise proportional to the local sensitivity which is equal to $\Delta_{max-min}$ where the values lie between upper and lower bounds $max$ and $min$. However, we must first obtain a noisy version $\tilde\Delta_{max-min}$ to the local sensitivity, and then add noise proportional to $\tilde\Delta_{max-min}$ to the axe.

The Published Axe is eventually given by:
\begin{eqnarray*}
a_{PUB}(t) &=& a_{TRUE}(0) + \sum_{i=1}^t \alpha^{+}(i) + \sum_{i=1}^t \alpha^{-}(i) \\
&=& a_{TRUE}(t) + \Theta^{+}(t) + \Theta^{-}(t) \\
\end{eqnarray*}
where $\Theta^{+}(t)\mathrm{ ,}\Theta^{+}(t)$ are the cumulative shocks.

To get better efficiency we also split the time-grid into buckets, each long $B$ days:
\begin{eqnarray*}
t &=& d(t)B + c(t) \\
c(t) &=& t \Mod{B} \\
d(t) &=& \frac{t - c(t)}{B}
\end{eqnarray*}

and define the cumulative shocks $\Theta^{+}(t)\mathrm{ ,}\Theta^{+}(t)$ as follows:

\begin{eqnarray*}
\Theta^{+}\big(t\big) &=& \Theta^{+}\big(p(t)\big) + \theta^{+}\big(t\big)  \\
\Theta^{-}\big(t\big) &=& \Theta^{-}\big(p(t)\big) + \theta^{-}\big(t\big)  \\
\end{eqnarray*}
where $T>B$ is a reset period and $p(t)$ is the start of the current B-bucket:
\begin{eqnarray*}
p(t) &=& d(t)B
\end{eqnarray*}

The random shocks inside a T-period are given by sensitivity which is the difference between maximum and minimum change of the True Axe:
\begin{eqnarray*}
\theta^{+}\big(t\big) &\sim& \textrm{Lap} \left( \frac{\Big|\max_{i\in [p(t),t]} - \min_{i\in [p(t),t]} \Big|}{\epsilon}\right)\\
\theta^{-}\big(t\big) &\sim& \textrm{Lap} \left( \frac{\Big|\max_{i\in [p(t),t]} - \min_{i\in [p(t),t]} \Big|}{\epsilon}\right)\\
\end{eqnarray*}

See Algorithm~\ref{Algorithm1} in the supplementary material for a concrete description of our algorithm.

\begin{theorem}\label{them}
Let $0<\delta<1$ and $\epsilon>0$. The continual aggregator mechanism is $2\epsilon$-differentially private. Furthermore, for each $t \in \mathbb{N}$, the  mechanism with block size $B$ is $(O(\frac{1}{\epsilon}\cdot\sqrt{\Delta\cdot (T/B+B)}\cdot\log \frac{1}{\delta}),\delta$)-useful at time $t$ out of the $T$ time steps.  
\end{theorem}
%(  1/{poly(t)} )

\begin{proof}
We will use the term item to refer to an integer in the stream $\sigma$. We let $\sum_{k=I}^j \sigma(k)$ to denote a partial sum involving items $i$ through item $j$. We start by observing that each item $\sigma(t)$ in the published stream $\alpha_{PUB}(t)$ appears in at most two noisy partial sums: at most one of the $\beta$’s and at
most one of the $\alpha$’s. In particular, let $d(t) = \frac{t}{B}$, then $\sigma(t)$ appears in only $\beta(d(t))$ and $\alpha(t)$. That said, the aggregator mechanism preserves $2\epsilon$-differential privacy. This means that if we change $\sigma(t)$, at
most $2$ noisy partial sums will be affected.

Next we focus on the utility: Observe that at any time $t = d(t)B + c(t)$
where $d(t), c(t) \in \mathbb{Z}$ and $0 \leq c(t) < B$, the error of $\alpha_{PUB}(t)$ includes the sum of $K =  d(t) + c(t)$ independent
Laplacian distributions $Lap(\Delta/\epsilon)$. The estimated aggregator at any time $t$ is the sum of at most $\lfloor t/B \rfloor+B$ partial sums. Followed from Corollary~\ref{cor}, since $t/B \leq K \leq (t/B + B)$, at time T, the Aggregator Mechanism is $(O(\frac{1}{\epsilon}\cdot\sqrt{\Delta\cdot (T/B+B)}\cdot\log \frac{1}{\delta}),\delta$)-useful at time t .

\end{proof}

Suppose a mechanism M adds $Lap( \Delta/\epsilon)$ noise to every partial sum before releasing it. In M, each item in the stream appears in at most $x$ partial sums, and each estimated aggregator is the sum of at most $y$ partial sums. Then, the mechanism M achieves $x\epsilon$-differential privacy. Moreover from Corollary~\ref{cor} the error is $O(\frac{\sqrt{\Delta\cdot y}}{\epsilon} )$ with high probability. Alternatively, to achieve $\epsilon$-differential privacy, one can scale appropriately by having $\epsilon’ = \epsilon/x$. Now if the mechanism instead adds  $Lap( \Delta/\epsilon')$  noise to each partial sum we achieve $\epsilon$-differential privacy, and $O(\frac{ x \sqrt{\Delta\cdot y}}{\epsilon }) $ error with high probability.

If we let $B =\sqrt{T}$, then the estimated aggregator at any time is the sum of at most $2B$ noisy partial sums. The
error is roughly $O(\frac{ T^{1/4}\sqrt{\Delta}}{\epsilon})$ with high probability.

\subsection{Binary Mechanism with Less Error}

In Algorithm \ref{Algorithm2} of the supplementary material we present a version of our mechanism which incurs a smaller error. In Algorithm ~\ref{Algorithm2} the estimated aggregator is the sum of at most $\log T$ noisy partial sums and each partial sum has an independent Laplace noise of $Lap(\frac{\log T\cdot\Delta}{\epsilon})$. That said, from Corollary~\ref{cor} we get the following: 

\begin{theorem}[Utility]\label{them2}
The continual aggregator mechanism in Algorithm \ref{Algorithm2} is $(O(\frac{1}{\epsilon}\cdot (\log T)\cdot\sqrt{\Delta\cdot \log t}\cdot\log \frac{1}{\delta}),\delta$)-useful at time $t$ out of the $T$ time steps.  
\end{theorem}

\section{Atlas-X Performance}\label{sec:exp}

We have implemented our differentially private continual Aggregator mechanism in the production (in Python code) of the internal brokerage platform. The production system operates on a Linux machine based on python 3.7 (256GB memory). The algorithm runs daily to obfuscate the bank's axe list for three different regions (USA, Europe and Asia). The obfuscated axe is then sent to roughly 60 selected clients (hedge funds), which all receive the same axe list. The rest of this section describes the results of several experiments performed with the monte carlo simulation engine described in Section \ref{SectionAxeSimulation} of the supplementary material based on real inventory data.

\begin{figure*}[!htb]
\minipage{0.32\textwidth}
  \includegraphics[width=\linewidth]{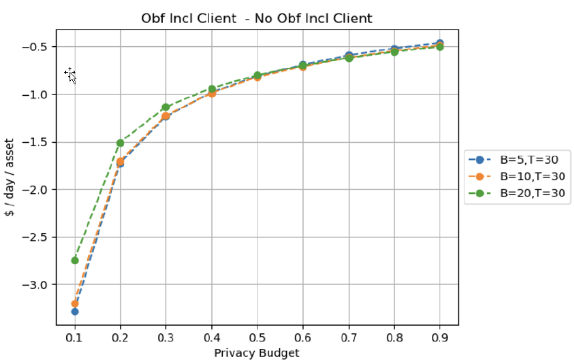}
  \caption{Expected inventory P\&L difference (Y-axis) between the case in which the bank publishes the DP obfuscated axe including the most concentrated client versus those calculated with the true (un-obfuscated) axe, measured in dollar per day per asset and calculated for different privacy budgets $\epsilon$ (X-axis) as well as obfuscation parameters.}\label{PLVsNoObfInclClient}
\endminipage\hfill
\minipage{0.32\textwidth}
  \includegraphics[width=\linewidth]{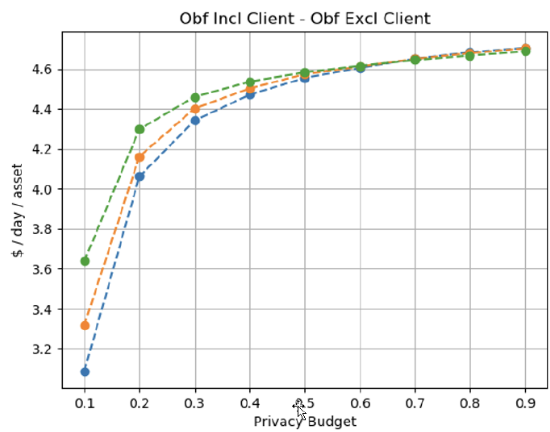}
 \caption{Expected inventory P\&L difference (Y-axis) between publishing the obfuscated axe with and without the concentrated client, respectively, measured in dollar per day per asset and calculated for different privacy budgets $\epsilon$ (X-axis) as well as obfuscation parameters.}\label{PLVsObfExclClient}
\endminipage\hfill
\minipage{0.32\textwidth}%
  \includegraphics[width=\linewidth]{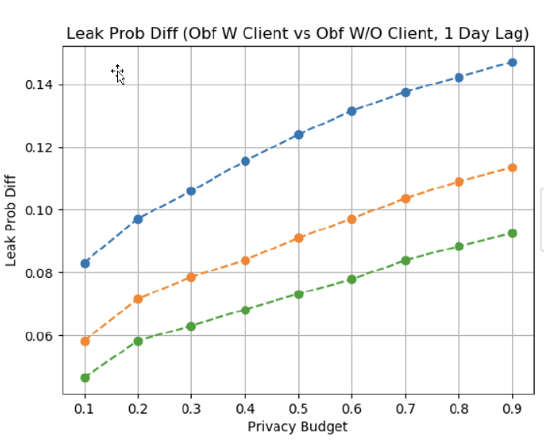}
   \caption{Expected Leak Probability difference (Y-axis) between publishing the
DP-obfuscated axe including the most concentrated client versus those excluding it, calculated for different privacy budgets $\epsilon$ (X-axis) and with a lag of 1 day.}\label{LeakProb-1d}
\endminipage
\end{figure*}

We first turn to the important question of how to determine the obfuscation model parameters discussed in Section \ref{sec:algo}. To do that, we have, instead, run several monte carlo simulations (using the methodology described in Section \ref{SectionAxeSimulation} of the supplementary material) for a grid of obfuscation parameters $\{(\epsilon, T,  B)\}$. We have then measured the obfuscation statistics described in Section \ref{SecObfuscationMetrics}, namely the Expected P\&L / Leakage Probability and discussed the results with the trading desk who were took the final decision (in the supplementary material we also measure the Over Axe Frequency and Worst Case Cost). Moreover, given that our application revolves around high-frequency trading, signals spanning a few days pose a significant threat if revealed. We have judiciously chosen a time parameter of T=30 days.

Our analysis was based on the selection of the most concentrated client, whose positions were effectively driving our axe on many assets. In particular, Fig. \ref{PLVsNoObfInclClient} shows the expected inventory P\&L difference between the case in which the bank publishes the DP obfuscated axe including the most concentrated client versus those calculated with the true (un-obfuscated) axe. As expected the results, calculated for different privacy budgets as well as obfuscation parameters, show that publishing the true axe is always advantageous from a P\&L perspective. The expected loss with the DP budget chosen in production ($\epsilon = 0.3$) is roughly 1\$ (daily, for each asset we publish an axe for). The P\&L difference increases with $\epsilon$ and flattens for larger values of $\epsilon$.

Fig. \ref{PLVsObfExclClient}, instead, compares the expected P\&L difference between publishing the obfuscated axe with and without the concentrated client, respectively, using the same model parameters. The results, show that there is an average 4.4\$ P\&L increase when including the concentrated client (per asset, per day) using the DP budget chosen in production. Again, the P\&L difference increases with $\epsilon$ and flattens for larger values.

The $P\&L$ estimates above were then compared with the expected Leakage Probability using the same model parameters. Fig. \ref{LeakProb-1d}, \ref{LeakProb-1w} and Fig. \ref{LeakProb-2w} (Fig. \ref{LeakProb-2w} in  the supplementary material) show the difference in the estimated Leakage Probability between publishing the obfuscated axe with and without the most concentrated client for time lags of 1-day / 1-week / 2-weeks, respectively. The results show that, for the DP budget chosen in production, there is an increase of $\sim$ 6\% in Leak Probability when including the most concentrated client over a 1-day lag, and $\sim$ 3\%  for 1-week or-2 weeks lags instead. That means that only in $3\%$ of the cases our published axe leaks information about the trading activity (direction) of the concentrated client over a 1-week or 2-weeks lags. This was considered as acceptable by the trading desk and the final model parameters used in production, corresponding to $(\epsilon = 0.3, T=30,  B=20)$.

%\begin{figure}[htb!]
%  \centering
%  \subfloat[Expected inventory P\&L difference (Y-axis) between the case in which the bank publishes the DP obfuscated axe including the most concentrated client versus those calculated with the true (un-obfuscated) axe, measured in dollar per day per asset and calculated for different privacy budgets $\epsilon$ (X-axis) as well as obfuscation parameters.]{\includegraphics[width=0.5\textwidth]{images/pl_obf-incl-client_noobf-incl-client.png}\label{PLVsNoObfInclClient}}
%  \hfill
%  \subfloat[Expected inventory P\&L difference (Y-axis) between publishing the obfuscated axe with and without the concentrated client, respectively, measured in dollar per day per asset and calculated for different privacy budgets $\epsilon$ (X-axis) as well as obfuscation parameters]{\includegraphics[width=0.5\textwidth]{images/pl_obf-incl-client_obf-excl-client.png}\label{PLVsObfExclClient}}
%  \caption{Expected inventory P\&L difference under different scenarios.}
%\end{figure}

%second
\begin{figure*}[!htb]
\minipage{0.43\textwidth}
  \includegraphics[width=\linewidth]{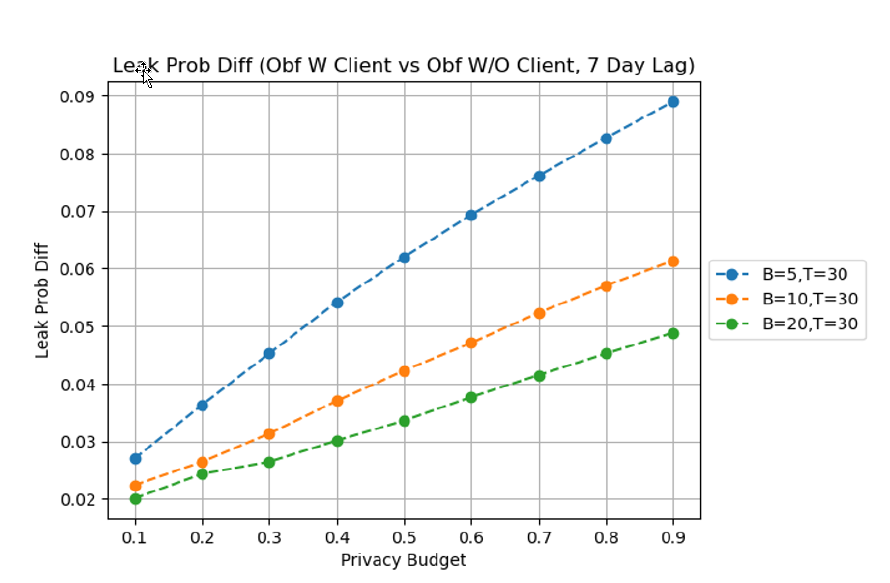}
  \caption{Expected Leak Probability difference (Y-axis) between publishing the
DP-obfuscated axe including the most concentrated client versus those excluding it, calculated for different privacy budgets $\epsilon$ (X-axis) and with a lag of 1 week.}\label{LeakProb-1w}
\endminipage\hfill
\minipage{0.48\textwidth}
  \includegraphics[width=\linewidth]{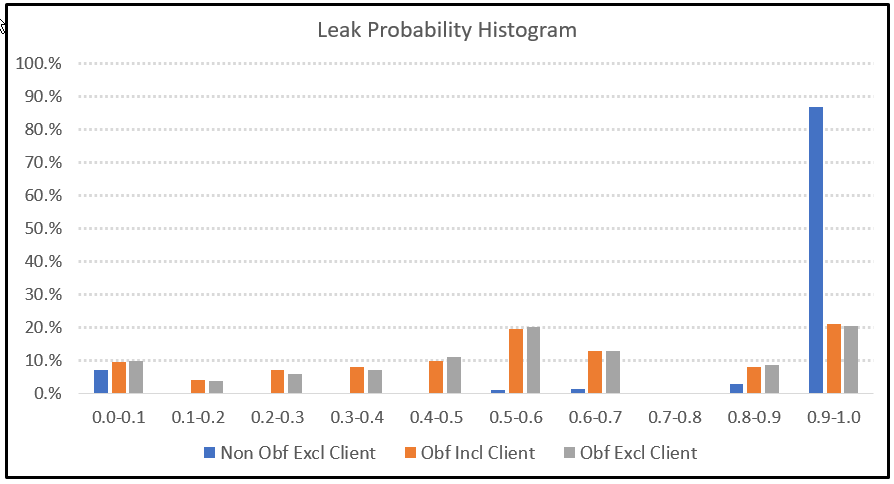}
  \caption{Histogram of simulated Leak Probabilities, with 2-weeks lag. The X-axis corresponds to the distribution deciles. The Y-axis reports the histogram frequencies for three different simulation assumptions: (a - Non Obf Excl Client) the Bank publishes the un-obfuscated axe, calculated without keeping the client's positions into account, (b - Obf Incl Client) the Bank publishes the obfuscated axe, including the client's positions, (c - Obf Excl Client) the Bank publishes the obfuscated axe, excluding the client's positions.}\label{LeakProbHist}
\endminipage\hfill
\end{figure*}

In a different experiment, we have analyzed a population of 600 assets traded by a highly concentrated client, for a period of one year based on real data, performing monte carlo simulations using the production DP parameters. Fig. \ref{LeakProbHist} shows the simulated histogram of the 2-weeks Leakage Probability in three different scenarios: (a) The bank publishes the un-obfuscated axe, calculated without keeping the client's positions into account. (b) The bank publishes the obfuscated axe, including the client's positions. (c) The bank publishes the obfuscated axe, excluding the client's positions.
The histogram shows, when one compares the data for cases (b) and (c), that the obfuscated axes calculated with and without the client's positions are indistinguishable. 

Please also notice that the strategy (a) of publishing the axe without DP obfuscation but calculated w/o the concentrated client positions, actually results in a Leakage Probability between 90\% and 100\% for the majority of the assets simulated (80\%). This means that, even when removing a client contribution from the axe calculation, it can often happen that the published axe moves in the same direction as the client’s trading activity. This is not a case of information leakage, but instead one in which hedge funds are effectively trading along similar strategies. It is a phenomenon sometimes called ”herding behaviour”. As discussed, it is nevertheless very important to publish an axe calculated using all clients positions, no matter how concentrated, because failure to do so would expose the bank to P\&L losses. In the supplementary material we present more results.

\section{Conclusion}

The Axe Inventory offering is a fundamental service in the financial world. In this work, we introduce a differential private mechanism under continual observation to obfuscate the axe inventory daily while hiding the client trading activity. Along the way we propose a continual differential private aggregator for streams of numbers. Our system is successfully running live in production in three different regions (USA,
Europe and Asia) at a major USA bank and, to our knowledge, it is the first differential privacy solution to be deployed in the financial sector. 
%%%%%%%%%%%%%%%%%%%%%%%%%%%%%%%%%%%%%%%%%%%%%%%%%%%%%%%%%%%%%%%%%%%%%%%%

%%% The acknowledgments section is defined using the "acks" environment
%%% (rather than an unnumbered section). The use of this environment 
%%% ensures the proper identification of the section in the article 
%%% metadata as well as the consistent spelling of the heading.

\begin{acks}
 This paper was prepared in part for information purposes by the Artificial Intelligence Research group of JPMorgan Chase \& Co, AlgoCRYPT CoE and its affiliates (``JP Morgan''), and is not a product of the Research Department of JP Morgan. JP Morgan makes no representation and warranty whatsoever and disclaims all liability, for the completeness, accuracy or reliability of the information contained herein. This document is not intended as investment research or investment advice, or a recommendation, offer or solicitation for the purchase or sale of any security, financial instrument, financial product or service, or to be used in any way for evaluating the merits of participating in any transaction, and shall not constitute a solicitation under any jurisdiction or to any person, if such solicitation under such jurisdiction or to such person would be unlawful. 2024 JP Morgan Chase \& Co. All rights reserved.
\end{acks}

%%%%%%%%%%%%%%%%%%%%%%%%%%%%%%%%%%%%%%%%%%%%%%%%%%%%%%%%%%%%%%%%%%%%%%%%

%%% The next two lines define, first, the bibliography style to be 
%%% applied, and, second, the bibliography file to be used.

\bibliographystyle{ACM-Reference-Format} 
\bibliography{AAMAS_2024_sample}

\appendix

\section{Financing Activities of a Prime Desk}\label{FinancingActivities}

In this section  we introduce the financial concepts and jargon used for the real use case in the the paper. A concise summary of the axe inventory use case is that the publishing bank aggregates its internal firm inventory of long (buy) and short (sell) trades and then provides these offerings to its customers in order to equalize their long and short aggregated positions with regard to a given asset.

\noindent{\bf Outline:} We define the profit and losses incurred by ``long'' and ``short'' positions (which refer to buying and selling respectively, and are defined in detail below), highlighting the importance of hedging costs for the the bank (``funding'' and ``borrow'' rates). We then demonstrate how banks reduce their hedging costs via a process known as ``internalization'' and how we can reduce such costs by enticing clients to trade via axe lists. Lastly, we describe the implications of sharing axe lists among clients and how axe lists can leak information about the trading activity of clients with large (``concentrated'') positions.\\

In order to maintain a level of simplicity, we have disregarded the formalization of complexities regarding the execution of real trades that are not central to the main themes of this paper.

The Prime desk of a broker-dealer bank facilitates clients' trading activities. Clients are generally hedge funds willing to profit on the price movement of assets belonging to a large universe, encompassing Equity / Fixed Income and Commodity instruments.

When a client is betting on a given asset to appreciate in the future, he/she can execute a ``long'' trade i.e. buy the asset with the idea of selling it later to capitalize on the price increase. To do that, the client needs to raise cash to cover the initial purchase. The client's Profit \& Loss (P\&L) after the asset is sold back in the market is given by:
\begin{eqnarray}
P\&L = N \big( P(t_{E}) - P(t_{S})\big) - N P(t_{S}) r_F (t_{E} - t_{S}) \label{LongPL}
\end{eqnarray}
where $t_{S}$ and $t_{E}$ are the trade's inception(Start)/termination(End) times, respectively. $N$ is the number of shares/units traded, $P(t)$ is the price of the asset at time $t$ and $r_F$ is the ``funding rate'' i.e. the interest rate paid to borrow the cash required to buy the asset.

If, instead, the client's view is that an asset will lose value in the future, he/she can put on a ``short'' trade i.e. sell the asset to then buy it back when its price has decreased in the market. To do so, the client needs to borrow the asset to cover the initial sale. In this case the client's P\&L after the trade is terminated is:
\begin{eqnarray}
P\&L = N \big( P(t_{S}) - P(t_{E})\big) - N P(t_{S}) r_B (t_{E} - t_{S}) 
\label{ShortPL}
\end{eqnarray}
where $r_B$ is the ``borrow rate'', that is the interest paid to borrow the asset.

The first term in Eq. \ref{LongPL} and \ref{ShortPL} is the ``risky'' P\&L, called in such a way because it is affected by the asset price movements. The second term is the funding/borrowing cost.

In general, the ``funding rate'' $r_F$ depends on the clients' credit-worthiness as well as the quality of the asset. It measures the ability of an institution to raise cash and it is in general not very volatile. The ``borrow rate'' $r_B$, instead, depends much more on the asset borrowed and can be very volatile when the underlying asset is under market stress.

A Prime desk facilitates the aforementioned activities by allowing its clients to execute ``leveraged'' trades, which means supplying clients with the required cash or assets to initiate their long/short trades, in exchange for a fee. When it does so, the Prime desk needs to hedge the corresponding risky P\&L and minimize the funding/borrowing costs of the hedge trades it needs to execute.

If the bank was to hedge each single client trade independently, something known as ``matched-book'' approach, it would then incur high costs due to the need - for each trade - to either (1) raise cash to allow a long trade, or (2) borrow the assets necessary for the short trade.

A better alternative for the bank is to utilize the cash/asset reserves on its ``balance sheet'' (See an example below).
A broker-dealer bank generally holds a balance sheet composed of a large number of positions of different types belonging on a large asset universe. For instance, some of them might be hedges for derivatives contracts, others might be inventory kept for market-making purposes while others are client positions that can be utilized (``re-hypothecated''). In particular, if the bank has a large enough balance sheet it will likely be able to match long and short positions on a given asset directly from its balance sheet, hence saving the corresponding financing or borrowing costs that would be present using a ``matched-book'' strategy. This process is known as ``internalization'' and is at the core of how a Prime desk functions.

We indicate the quantity of a position of type $p$ on a given asset at time $t$ as $x_{p}(t)$\footnote{The definition of ``position type'' is loose, it can refer to the individual  executed trades as well to a coarser categorization. For instance, in the example below, the position types used, Client / Delta One / Exotics, are representative of trading desks.}. A positive position corresponds to long trades (the bank bought the asset) while negative positions correspond to short trades (the bank borrowed then sold the asset). The net/internalized position is then given by the sum over all position types:
\begin{eqnarray}
x(t) &=& \sum_{p} x_{p}(t)
\end{eqnarray}

For instance, the table below represents a hypothetical balance sheet for three assets $A$, $B$ and $C$:\\

\begin{tabular}{llrr}
\toprule
 Asset   & Type   &   Long Quantity &   Short Quantity \\
\midrule
 A       & Client           &             100 &                0 \\
 A       & Delta One   &               0 &             -100 \\
 B       & Client  &             100 &                0 \\
 B       & Delta One   &              50 &              -50 \\
 B       & Exotics     &              20 &              -10 \\
 C       & Client   &               0 &              -30 \\
\bottomrule
\end{tabular}\\\\

After internalizing the positions, the net position vector $x(t)$ would then be:\\

\begin{tabular}{lr}
\toprule
 Asset   &   Net Quantity \\
\midrule
 A       &              0 \\
 B       &            110 \\
 C       &            -30 \\
\bottomrule
\end{tabular}\\\\

In this example, asset $A$ is perfectly internalized having a net position equal to zero. In such a case the bank would incur no funding or borrowing cost as there are no hedge trades to be kept on balance sheet. The net positions for assets $B$ and $C$, instead, would be positive/negative. This means that the bank would incur funding/borrowing costs, respectively.

Whenever the net position is different from zero, $x(t) \ne 0$, the bank is hit by an inventory cost which accrues daily.
 When $x(t) > 0$ corresponding to long net positions, the bank needs to raise cash to cover the purchase of the shares of the asset that have not been internalized. We indicate, again, the interest paid on the borrowed cash with the funding rate $r_F(t)$. On the other hand, for short net positions which correspond to  $x(t) < 0$, the bank needs to borrow the (non internalized) asset shares at a cost measured by the borrowing rate $r_B(t)$. 

If the internalization was perfect we would have $x(t) = 0$ for all assets on balance sheet. This is an ideal situation because the bank would run its activities with no financing or borrowing costs whatsoever. In practice, perfect internalization is never achievable and hedge trades need to be executed. Referring to Case A in Figure \ref{FigHedging}, such hedges are generally put in place with another bank and are expensive because of the need to pay fees to the counterparty. Hence it makes sense for the bank to publish (via email or other communication channels) a list of assets / quantities on which it is willing to trade at a discount for the sole purpose of achieving better internalization. Such a list is generally called an ``axe'' and we indicate it with $a_{PUB}(t)$. If the bank can find a matching client who is willing to sell the same amount of shares (Case B in Figure  \ref{FigHedging}), it can replace the expensive hedge with a client trade where no fees are paid, thus saving costs. An axe trade is beneficial for the client as well, because of the lower cost charged by the bank. Banks, therefore, put efforts in locating such client matches.

\begin{figure}
  \centering
  \includegraphics[width=0.35\textwidth]{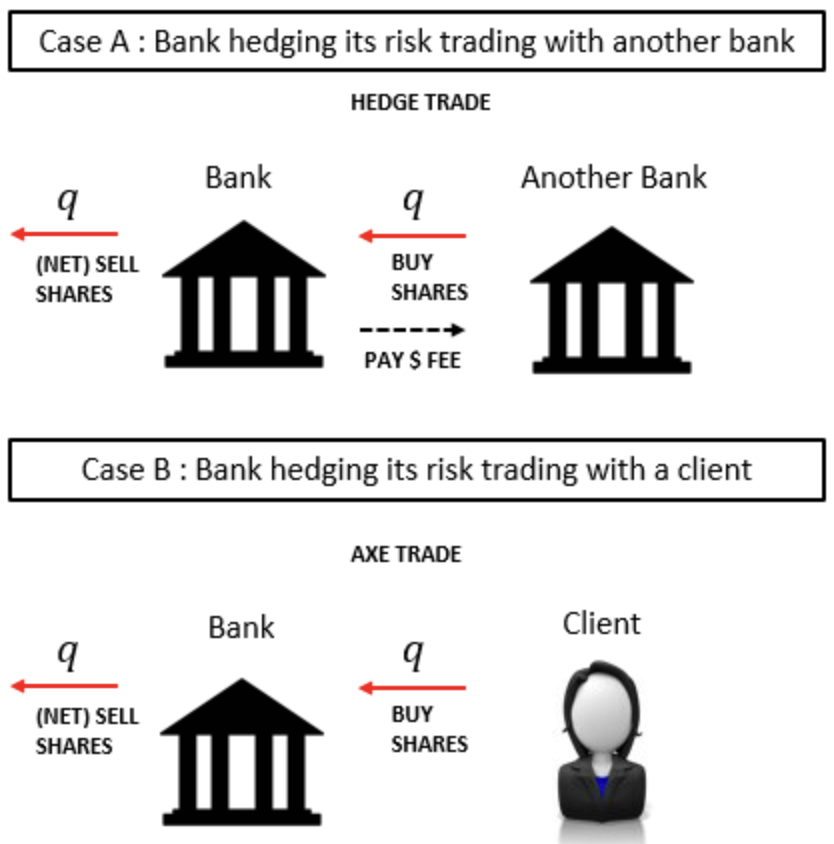}
  \caption{Illustration of the difference between a hedge traded with another bank (Case A) and an axe traded with a client (Case B). The arrows indicate the direction in which the asset (red) or cash fee moves (black) for instance, when the bank is selling an asset, the corresponding arrow points out of the bank. The bank is a net seller of $q$ shares of an asset (red arrows on the left) and to neutralize its risk it has to buy the same amount of shares (red arrows on the right). In Case A the bank transacts with another bank, which will charge fees for the service provided (black arrow). In Case B, instead, the bank managed to find a client willing to enter into an axe trade and the hedge will be costless (no fees).}\label{FigHedging}
\end{figure}

When the published axe is not obfuscated, it corresponds to the opposite of the true balance sheet positions and is given by $a_{TRUE}(t) = - x(t)$\footnote{The minus follows the convention according to which a long  published axe (positive quantity) corresponds to a short (negative) net balance sheet position, and vice versa.}.
However, this is unsatisfactory:  (1) The bank reveals its aggregate inventory, signaling whether it is a net buyer or seller of the stock. Clearly, the bank would like to prevent this kind of leakage; (2) Releasing aggregate/summed information about clients in the list may seem harmless. However, such statistical data can expose sensitive information about a ``concentrated'' client, i.e. one which holds a large portion of the assets in the bank's balance sheet. The trading activity of a concentrated client can be revealed to other clients even though the list is anonymized and includes aggregated information of all clients. In particular, continually updating the published statistics over time can give even more leverage to the adversary and result in more privacy concerns. For example, when a concentrated client executes a large trade, such a move will be reflected in the published axe and reveal a trading strategy that should instead be kept secret.

When a client receives a published axe $a_{PUB}(t)$ from the bank, it can opt to execute a quantity $a_{HIT}(t)$ of it. If the published axe is positive, the client can opt to enter long / buy trades ($a_{HIT}(t) \ge 0$). Otherwise, if the published axe is negative the client can sell the asset short ($a_{HIT}(t) \le 0$). Also, the client cannot trade more than the communicated axe quantity. Such constraints can be expressed as follows:
\begin{eqnarray}
\sgn(a_{HIT}(t)) &=& \sgn(a_{PUB}(t)) \\
|a_{HIT}(t)| &\le&|a_{PUB}(t)| \notag
\end{eqnarray}

When an axe trade happens, the corresponding net positions on the bank's balance sheet change from $x(t)$ to $x(t) + a_{HIT}(t)$ and a P\&L profit is realized for the bank (See Section \ref{SecPLMetric} for details on the P\&L).

As discussed in this paper, the problem with publishing the true axe $a_{TRUE}(t)$ is that its dynamics track the bank's balance sheet's and, worse, the trading activity of large concentrated clients (whose positions drive the bank's balance sheet for specific assets). This can leak valuable information to external observers, namely all the clients receiving the axe list, and can have serious consequences from a risk or reputational perspective.

\begin{algorithm*}[ht!]
\KwIn{Stream $a_{TRUE} \in \mathbb{Z}$, privacy budget $\epsilon$ and time upper bound $T$}
\KwOut{At each time step $t$, output $a_{PUB}(t)$}

{\bf Initialization:} For all $t\in[T]$, compute two streams $\sigma^{+}(t)$ and $\sigma^{-}(t)$ representing the true axe differences: 
\begin{eqnarray*}
\sigma(t) &=& a_{TRUE}(t) - a_{TRUE}(t-1)\\
\sigma^{+}(t) &=& \sigma(t) \textrm{ if } \sigma(t) >0  \textrm{ else } \sigma^{+}(t)= 0 \\
\sigma^{-}(t) &=& \sigma(t) \textrm{ if } \sigma(t) <0  \textrm{ else } \sigma^{-}(t)=0
\end{eqnarray*}

such that:
\begin{eqnarray*}
a_{TRUE}(t) &=& a_{TRUE}(0) + \sum_{i=1}^t \sigma^{+}(i) + \sum_{i=1}^t \sigma^{-}(i)
\end{eqnarray*}

{\bf Obfuscation at time step $t\in T$ :} Split the time-grid $T$ into buckets, each long $B$ days such that $t = d(t)B + c(t)$ where $c(t) = t \Mod{B}$ and $d(t) = \frac{t - c(t)}{B}$.

If $c(t)\neq 0$ perturb each value as follows:
\begin{eqnarray*}
\alpha^{+}(t) &=& \sigma^{+}(t) + \theta^{+}\big(t-c(t),t\big) \\
\alpha^{-}(t) &=& \sigma^{-}(t) + \theta^{-}\big(t-c(t),t\big) 
\end{eqnarray*}
where: 
\begin{eqnarray*}
\theta^{+}\big(p(t), t\big) &\sim& \textrm{Lap} \left( \frac{\Big|\max_{i\in [p(t),t]} - \min_{i\in [p(t),t]} \Big|}{\epsilon}\right)\\
\theta^{-}\big(p(t), t\big) &\sim& \textrm{Lap} \left( \frac{\Big|\max_{i\in [p(t),t]}- \min_{i\in [p(t),t]} \Big|}{\epsilon}\right)
\end{eqnarray*}

For the case where $c(t)=0$
\begin{eqnarray*}
\beta^{+}\big(d(t)\big) &=&  \sum_{i=d(t)\cdot B}^t \sigma^{+}(i) + \theta^{+}\big((d(t)-1)B,t\big) \\
\beta^{-}\big(d(t)\big) &=&  \sum_{i=d(t)\cdot B}^t \sigma^{-}(i) + \theta^{-}\big((d(t)-1)B,t\big) 
\end{eqnarray*}

Compute the published axe as follows: 

\begin{eqnarray*}
a_{PUB}(t) &=& a_{TRUE}(0) + \sum_{i=1}^{d(t)} \big(\beta^{+}(i)+\beta^{-}(i)\big) + \sum_{i=d(t)B+1}^t \big(\alpha^{+}(i)+ \alpha^{-}(i)\big) \\
&=& a_{TRUE}(t)  + \sum_{i=1}^{d(t)} \Big(\theta^{+}\big((i-1)B,i\cdot B\big)
+\theta^{-}\big((i-1)B,i\cdot B\big)\Big) + \sum_{i=d(t)B+1}^t \Big(\theta^{+}\big(d(t)B,i\big)+ \theta^{-}\big(d(t)B,i\big)\Big) \\
\end{eqnarray*}

    \caption{\label{Algorithm1}}

\end{algorithm*}

\section{Axe Simulation}\label{SectionAxeSimulation}

To analyze the quality of various axe obfuscation strategies (both in terms of P\&L and information leakage), we have implemented a monte carlo axe simulation engine whose
logic  mimics the following sequence of trading actions executed every day by the bank's Prime desk (see Fig. \ref{AxeSimulation}):
\begin{itemize}
\item[-] The true (Pre) axe is calculated at start of day.
\item[-] An obfuscation strategy is applied to the true axe, producing a published axe that is presented to clients.
\item[-] Some of the published axe offers are accepted by clients, who trade with the Bank hence changing the true axe as well as the inventory P\&L.
\end{itemize}

The approach we have taken is to use our historical axe data as the backbone of our simulations and then keep into account the effect of simulated axe trades. 
To do that, we introduce the following axe quantities:
\begin{itemize}
\item[-] $a_{HIST}(t)$:  Historical axe, not simulated
\item[-] $a_{PRE}(t)$:  True axe at the start the trading day, keeps into account both the historical axe as well as axe trades generated in the previous simulation steps
\item[-] $a_{POST}(t)$:  True axe at the end of the trading day, keeps into the axe traded generated at the current simulation step
\item[-] $a_{PUB}(t)$:  Published axe
\item[-] $a_{HIT}(t)$:  Axe trades executed by clients
\end{itemize}

The monte carlo stepping is then implemented as follows:
\begin{eqnarray*}
a_{PRE}(t) &=& a_{HIST}(t) - \sum_{t' \in [t-H, t)} a_{HIT}(t') \\ \notag
a_{PUB}(t) &=& F\Big(a_{PRE}(t) \Big)\\ \notag
a_{HIT}(t) &=&  h \ a_{PUB}(t) \\ \notag
a_{POST}(t) &=& a_{HIST}(t) - \sum_{t' \in [t-H, t]} a_{HIT}(t')   \\ \notag
&=&a_{PRE}(t) - a_{HIT}(t)\\ \notag
\end{eqnarray*}

where:
\begin{itemize}
\item[-] $h$ is the \emph{axe hit ratio}, the constant proportion of executed axe trades versus published axe. In our simulations, we have used a hitting ratio $h\simeq 5\% / 10\%$ which is consistent with historical data.
\item[-] $H$ is the \emph{axe holding period}, indicating the duration of the axe trades. We use two weeks as the holding period.
\item[-] $F()$ is the chosen obfuscation strategy
\end{itemize}

\begin{figure}
  \centering
  \includegraphics[width=0.45\textwidth]{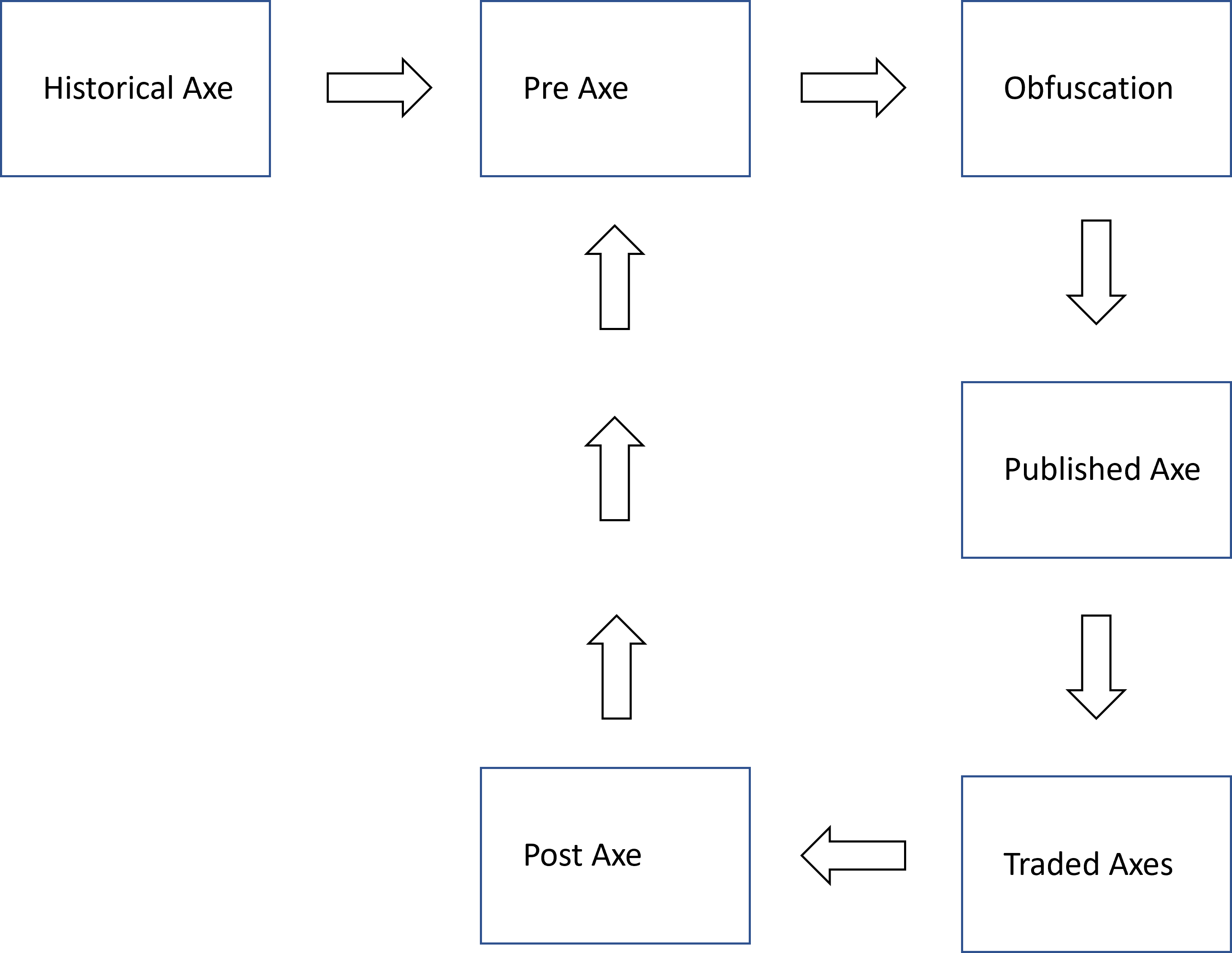}
  \caption{Axe Simulation Steps}\label{AxeSimulation}
\end{figure}

Using such monte carlo engine the authors have run several experiments and measured the effect of different DP parameters on metrics of interest, see Sections \ref{SecObfuscationMetrics} and \ref{sec:exp} of the main body of the paper.

\subsection{Over-Axe Frequency and Worst Case Cost}
We define the Over-Axe Quantity as the following Hinge loss:
\begin{eqnarray}
Q^{OA}(t) = \begin{cases}
        \max\Big(0, a_{PUB}(t) - a_{TRUE}(t)\\  \hspace{1cm} \Big[1 + \frac{r_B(t)}{r_F(t)}\Big] \Big) + 
       \max\Big(0, -a_{PUB}(t)\Big) \\\hspace{1cm} \textrm{if } a_{TRUE}(t) \ge 0  \\
        \max\Big(0, -a_{PUB}(t) + a_{TRUE}(t)\\\hspace{1cm} \Big[1 + \frac{r_F(t)}{r_B(t)}\Big] \Big)+ \max\Big(0, a_{PUB}(t)\Big) \\\hspace{1cm} \textrm{if } a_{TRUE}(t) < 0  
    \end{cases}
\end{eqnarray}
It measures how far the published axe is from the P\&L bounds given in Eq. \ref{PLBounds} (of the main body of the paper) and it is directly related to P\&L losses due to over-axing.
The Over-Axe Frequency is then given by:
\begin{eqnarray}
F^{OA}(t) &=& Pb\Big[Q^{OA}(t) > 0\Big]
\end{eqnarray}
It measures how often the published axe, if fully accepted by clients, would cause a negative inventory P\&L / loss for the bank.
As metric for the losses due to over-axing we use the Over-Axe Worst Case Cost, which is the cost incurred if the axe is fully hit every time the Over-Axe Quantity is different from zero and never hit otherwise. 
\begin{eqnarray}
C^{OA}(t) = E\bigg[ \begin{cases}
        Q^{OA}(t) r_B(t) & \textrm{if } a_{TRUE}(t) \ge 0  \\
        Q^{OA}(t) r_F(t) & \textrm{if } a_{TRUE}(t) < 0  
    \end{cases}\Bigg]
\end{eqnarray}
This is a very conservative estimator, as it corresponds to the worse-case scenario in which clients only accepts axes that cause a P\&L loss and never those producing a gain. Please notice, in this sense, that clients very rarely trade on the full published axe, see the discussion on the ``axe hit ratio'' in Section~\ref{SectionAxeSimulation}.

\section{Experimental Results Cont.}

Using the same model parameters, Fig. \ref{HistScenario} illustrates four simulated scenarios for the obfuscated  published axe for a given asset, together with the historical data for the bank's true axe  and the positions of a highly concentrated client. It can be noticed how the differential private published axe randomization makes any short-term inference, on whether the bank's true balance sheet positions or the client positions are increasing or decreasing, very challenging. At the same time the published axe follows, on average, the dynamics of the true axe. For such a reason, any axe trade executed in those scenarios would have produced a P\&L gain for the bank.

We eventually double-checked the production model parameters by analyzing the expected Over Axe Frequency and Worst Case Cost. Fig~\ref{OverAxeFreq} and \ref{OverAxeCost} report such metrics for the production model parameters. With the chosen parameters, the trading desk was happy to bear the risk of an Over Axe Frequency of $\sim$ 3.5\% and Over Axe Worst Case Cost of $\sim$ 6 \$  per day for each asset. Regarding this point, please also notice that over-axing is in general not a problem because (1) our algorithm very rarely
publishes axe quantities beyond the limits given by Eq.\ref{PLBounds} and (2) clients rarely accept the full axe but only a small proportion of it\footnote{As discussed in section \ref{SectionAxeSimulation} the historical axe hitting ratio is $\simeq 5 \%$ i.e. very low in practice.}.

\begin{figure*}
  \centering
  \includegraphics[width=0.9\textwidth]{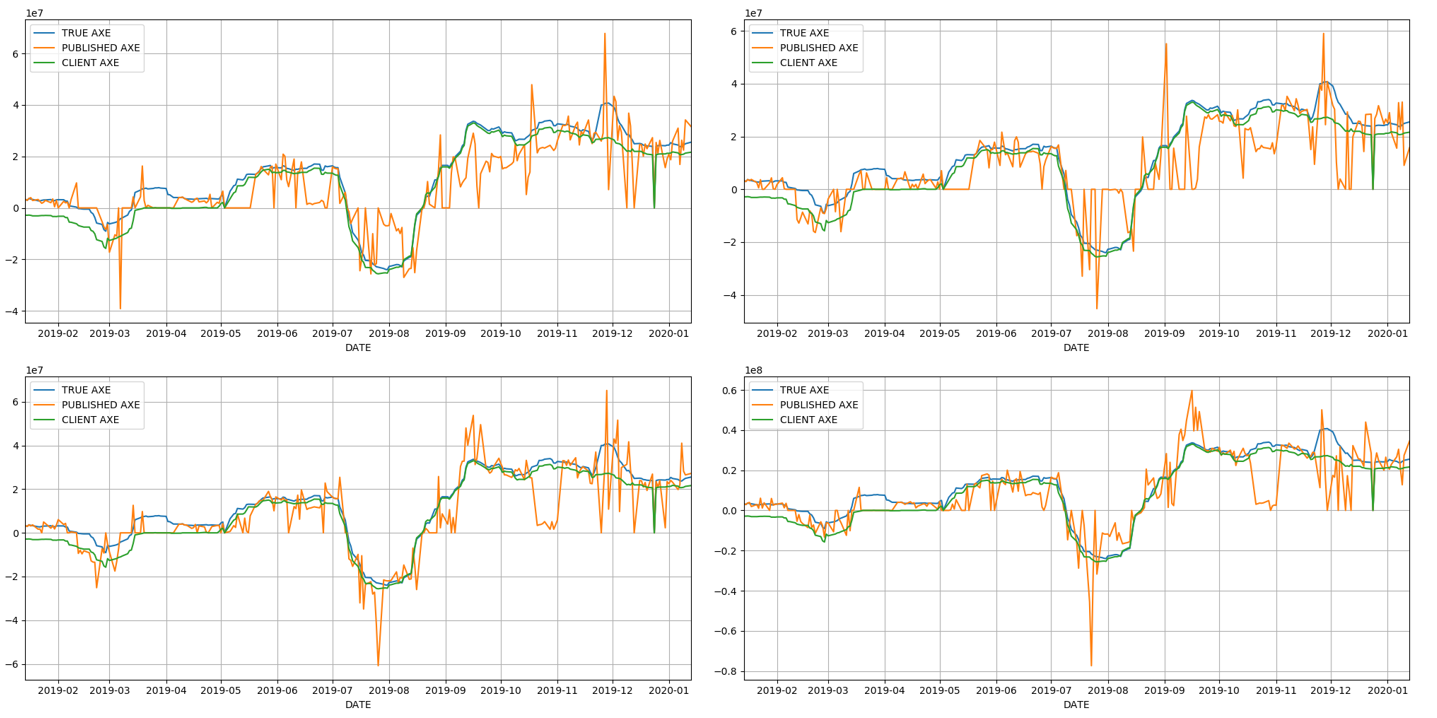}
  \caption{Four simulated scenarios for the obfuscated published axe (in orange color) for a given asset, together with the historical data for the bank's true axe (in blue color) and the positions of a highly concentrated client (in green color). The Y-axis refers to the axe quantity while the X-axis the observation date.}\label{HistScenario}
\end{figure*}

\begin{figure*}[!htb]
\minipage{0.47\textwidth}
  \includegraphics[width=\linewidth]{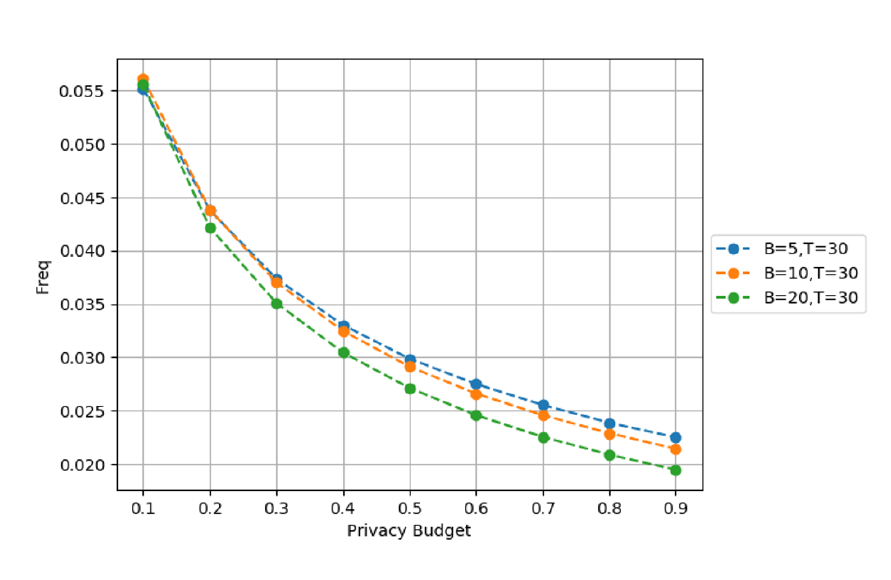}
  \caption{Expected Over Axe Frequency (Y-axis) calculated with monte carlo runs for different privacy budgets $\epsilon$ (X-axis) and obfuscation model parameters.}\label{OverAxeFreq}
\endminipage\hfill
\minipage{0.47\textwidth}
  \includegraphics[width=\linewidth]{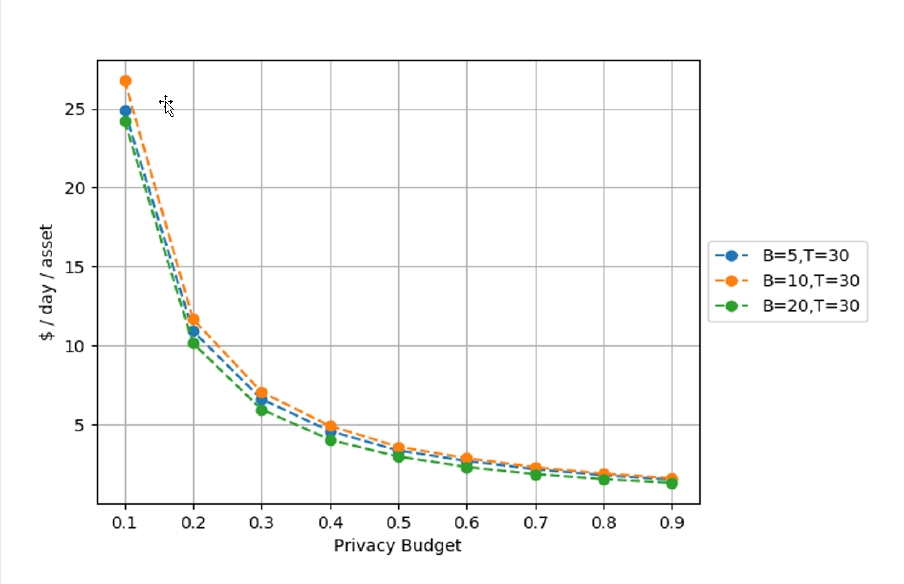}
  \caption{Expected Worst Case Over Axe Cost (Y-axis) calculated with monte carlo runs for different privacy budgets $\epsilon$ (X-axis) and obfuscation model parameters.}\label{OverAxeCost}
\endminipage\hfill
\end{figure*}

\begin{figure*}[!htb]
\minipage{0.47\textwidth}
  \includegraphics[width=\linewidth]{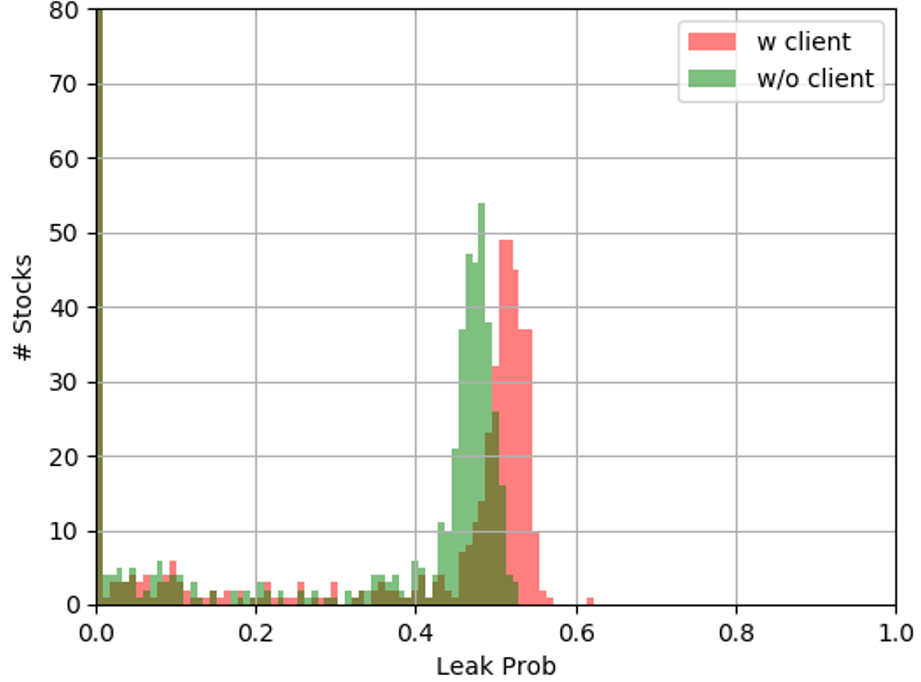}
  \caption{Histogram of simulated Leak Probabilities with 2-weeks lag (Y-axis), sampled with a concentrated client's positions included (red color) or excluded (green color) from the calculation of the published axe. The X-axis refers to the Leak Probability and the used Privacy Budget is $\epsilon = 0.9$. }\label{LeakProbHistComp09}
\endminipage\hfill
\minipage{0.47\textwidth}
  \includegraphics[width=\linewidth]{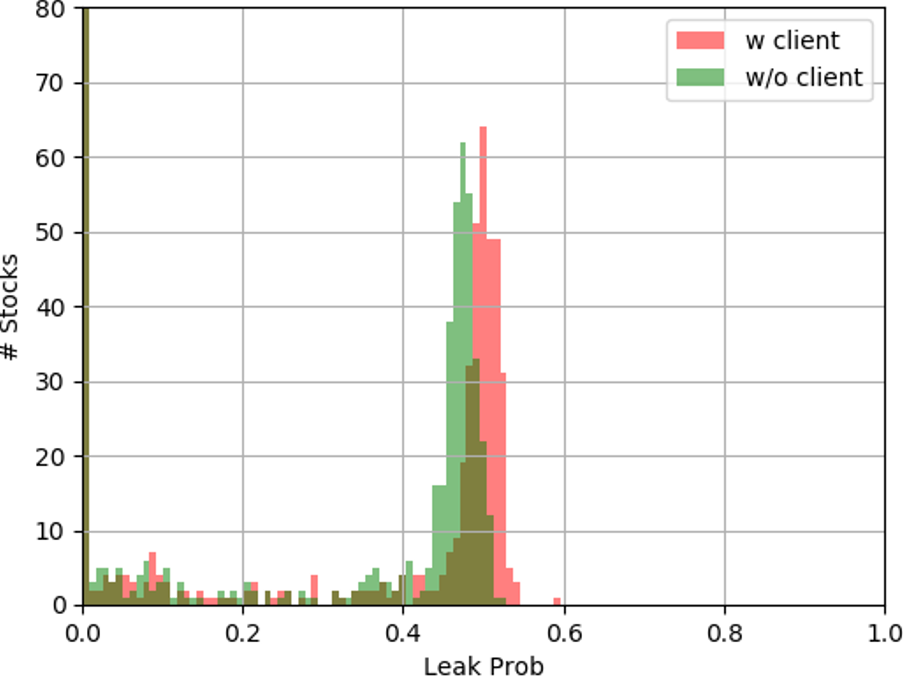}
  \caption{Histogram of simulated Leak Probabilities with 2-weeks lag (Y-axis), sampled with a concentrated client's positions included (red color) or excluded (green color) from the calculation of the published axe. The X-axis refers to the Leak Probability and the used Privacy Budget is $\epsilon = 0.5$. }\label{LeakProbHistComp05}
\endminipage\hfill
\end{figure*}

\begin{figure*}[!htb]
\minipage{0.47\textwidth}
  \includegraphics[width=\linewidth]{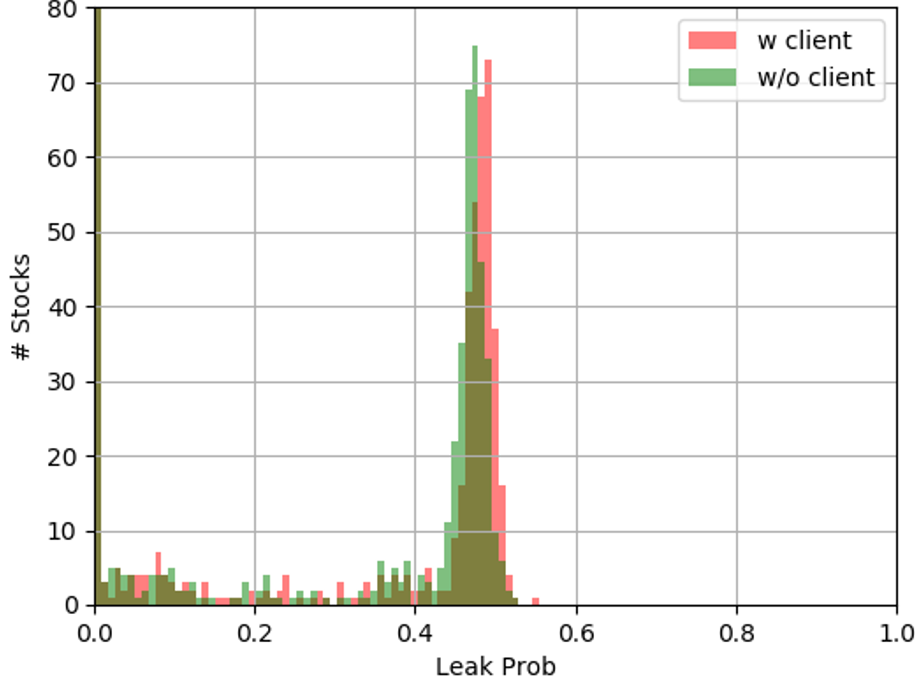}
  \caption{Histogram of simulated Leak Probabilities with 2-weeks lag (Y-axis), sampled with a concentrated client's positions included (red color) or excluded (green color) from the calculation of the published axe. The X-axis refers to the Leak Probability and the used Privacy Budget is $\epsilon = 0.1$. Compared to the results in Figure \ref{LeakProbHistComp09}, the lower $\epsilon$ make the two distribution closer.}\label{LeakProbHistComp01}
\endminipage\hfill
\end{figure*}

Using a different asset inventory, Figures \ref{LeakProbHistComp09}, \ref{LeakProbHistComp05} and \ref{LeakProbHistComp01} show the histograms of simulated Leak Probabilities (2-weeks lag), obtained using the full balance sheet data (red color) or discarding the positions belonging to a concentrated client (green color), for three different values of the Privacy Budget ($\epsilon = 0.9$ / $\epsilon = 0.5$ / $\epsilon = 0.1$, respectively). It can be noted how higher obfuscation make the two Leak Probability histograms closer.

Fig. \ref{SynthScenario} describes a simple experiment in which we generated a synthetic scenario in which a single client (1) contributes to the full axe for a given name and (2) increases his position gradually with a 10-fold increment in a period of two months. We report six random scenarios for the obfuscated axe, calculated with production DP parameters (See below for how they were determined), to give the reader a feeling of how obfuscation is affecting the published axe in such a simplified setting. We can note how the noise injected in the published axe is enough to make difficult for an external attacker to correctly guess the direction of change of the true balance sheet positions, at least over short time horizons. Over longer periods, the published axe follows the true axe.

\begin{figure*}
  \centering
  \includegraphics[width=0.9\textwidth]{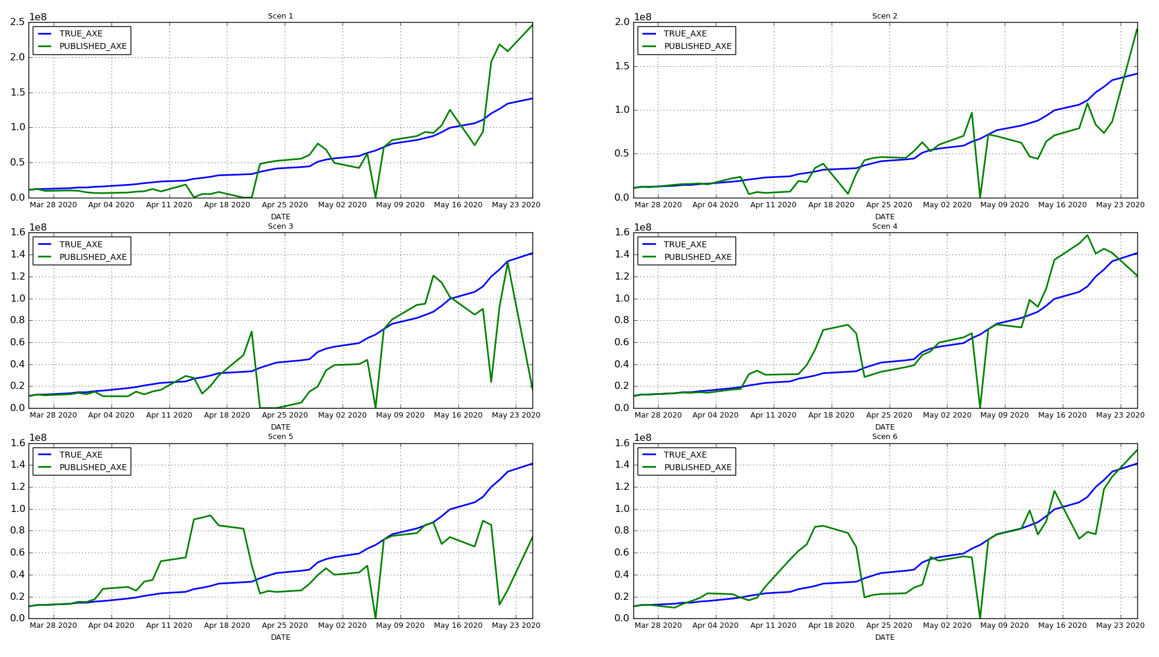}
  \caption{Monte carlo obfuscation scenarios generated with synthetic balance sheet data. We assume that a single client contributes to the bank's positions on a given asset, steadily increasing the balance sheet quantity over a period of two months. 
  The Y-axis refers to the axe quantity while the X-axis the observation date.
  The plots report six different randomized scenarios for the published axe (in green color) as well as the true axe (in blue color).}\label{SynthScenario}
\end{figure*}

\begin{figure*}[!htb]
\minipage{0.47\textwidth}
  \includegraphics[width=\linewidth]{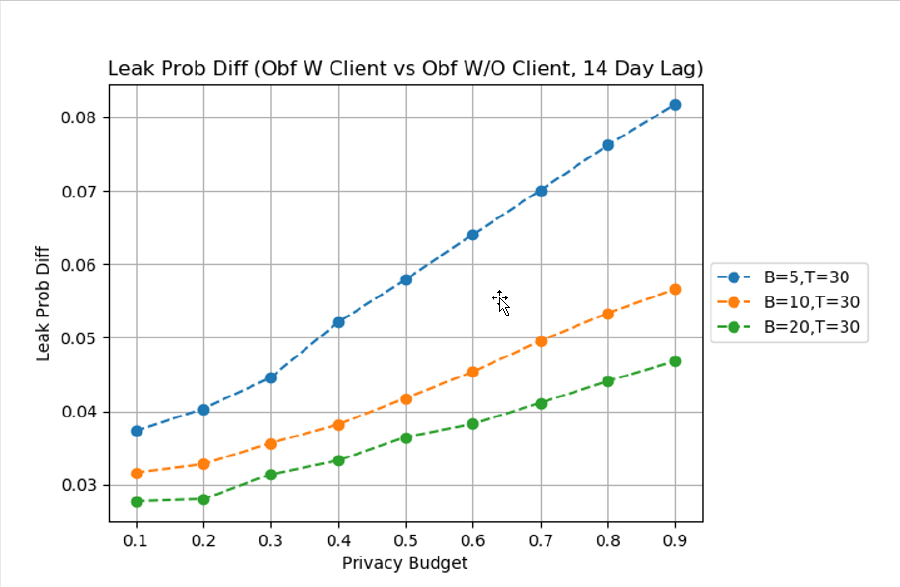}
  \caption{Expected Leak Probability difference (Y-axis) between publishing the
DP-obfuscated axe including the most concentrated client versus those excluding it, calculated for different privacy budgets $\epsilon$ (X-axis) and with a lag of 2 weeks.}\label{LeakProb-2w}
\endminipage\hfill
\end{figure*}

\begin{algorithm*}[ht!] 
\KwIn{Stream $a_{TRUE} \in \mathbb{Z}$, privacy budget $\epsilon$ and time upper bound $T$}
\KwOut{At each time step $t$, output $a_{PUB}(t)$}

{\bf Initialization:} For all $t\in[T]$, compute two streams $\sigma^{+}(t)$ and $\sigma^{-}(t)$ representing the true axe differences: 
\begin{eqnarray*}
\sigma(t) &=& a_{TRUE}(t) - a_{TRUE}(t-1)\\
\sigma^{+}(t) &=& \sigma(t) \textrm{ if } \sigma(t) >0  \textrm{ else } \sigma^{+}(t)= 0 \\
\sigma^{-}(t) &=& \sigma(t) \textrm{ if } \sigma(t) <0  \textrm{ else } \sigma^{-}(t)=0
\end{eqnarray*}

such that:
\begin{eqnarray*}
a_{TRUE}(t) &=& a_{TRUE}(0) + \sum_{i=1}^t \sigma^{+}(i) + \sum_{i=1}^t \sigma^{-}(i)
\end{eqnarray*}

{\bf Obfuscation at time step $t\in T$ :} Divide the time-grid $T$ into $\log T$ intervals and initialize $\rho_1,...,\rho_{\log T}$ to zero. Take the binary representation of $t$, denoted as $(t)_2$ and let $(t_i)_2$ denote the $i$-th bit of $(t)_2$.

Let $i$ be the least significant bit of $(t)_2$ for which $(t_i)_2=1$ (i.e., $i=LSB\{j:(t_j)_2\neq 0\}$) then

\begin{eqnarray*}
\rho^{+}_i &=& \sum_{j<i}\rho^{+}_j + \sigma^{+}(t) \\
\rho^{-}_i &=& \sum_{j<i}\rho^{-}_j + \sigma^{-}(t)
\end{eqnarray*}

Perturb each value as follows:

\begin{eqnarray*}
\alpha^{+}_i &=& \rho^{+}_i + \theta^{+}\big(t-\log T,t\big) \\
\alpha^{-}_i &=& \rho^{-}_i + \theta^{-}\big(t-\log T,t\big) 
\end{eqnarray*}

where: 
\begin{eqnarray*}
\theta^{+}\big(p(t), t\big) &\sim& \textrm{Lap} \left( \frac{\log T\Big|\max_{i\in [p(t),t]} - \min_{i\in [p(t),t]}\Big|}{\epsilon}\right)\\
\theta^{-}\big(p(t), t\big) &\sim& \textrm{Lap} \left( \frac{\log T\Big|\max_{i\in [p(t),t]} - \min_{i\in [p(t),t]}\Big|}{\epsilon}\right)
\end{eqnarray*}

Compute the published axe as follows: 

\begin{eqnarray*}
a_{PUB}(t) &=& a_{TRUE}(0) + \sum_{j:(t_j)_2= 1} \big(\alpha^{+}(j)+ \alpha^{-}(j)\big) \\
&=& a_{TRUE}(t) + \sum_{j:(t_j)_2= 1} \big(\theta^{+}(j)+ \theta^{-}(j)\big) \\
\end{eqnarray*}

    \caption{\label{Algorithm2}}

\end{algorithm*}
%%%%%%%%%%%%%%%%%%%%%%%%%%%%%%%%%%%%%%%%%%%%%%%%%%%%%%%%%%%%%%%%%%%%%%%%

\end{document}